\documentclass[reqno]{amsart}

\usepackage[top=1in, bottom=1in, left=1in, right=1in]{geometry}
\usepackage{tipa}
\usepackage{amssymb}
\usepackage{graphicx}
\usepackage{hyperref}
\usepackage{enumerate}
\usepackage{xcolor}
\usepackage{color}
\usepackage{mathrsfs}     
\usepackage{dsfont}
\usepackage{makecell}    
\usepackage[format=plain,justification=centering,singlelinecheck=false,]{caption}   
\usepackage[FIGTOPCAP]{subfigure} 

\usepackage{amsmath}    
\allowdisplaybreaks[4]

\def\be{\begin{equation}}
\def\ee{\end{equation}}

\newtheorem{theorem}{Theorem}[section]
\newtheorem{lemma}[theorem]{Lemma}
\newtheorem{prop}[theorem]{Proposition}

\newtheorem{remark}{Remark}[section]

\numberwithin{equation}{section}

\newcommand{\abs}[1]{\left|#1\right|}                     

\begin{document}

\title[Pricing VIX Derivatives With Free Stochastic Volatility Model]{Pricing VIX Derivatives With Free Stochastic Volatility Model}

\author[W. Lin]{Wei Lin}
\address{School of Mathematical Sciences , Zhejiang University, Hangzhou, 310027, People's Republic of China}
\email{weilin1991@zju.edu.cn; mathslin@126.com}
\thanks{}

\author[S. H. Li]{Shenghong Li}
\address{School of Mathematical Sciences , Zhejiang University, Hangzhou, 310027, People's Republic of China}
\email{shli@zju.edu.cn}

\author[S. Chern]{Shane Chern}
\address{School of Mathematical Sciences , Zhejiang University, Hangzhou, 310027, People's Republic of China}
\email{shanechern@zju.edu.cn; chenxiaohang92@gmail.com}


\date{}

\dedicatory{}


\begin{abstract}
%
In this paper, we relax the power parameter of instantaneous variance and develop a free stochastic volatility plus jumps model that generalise the Heston model and 3/2 model as special cases. This model has two distinctive features. First, we do not restrict the newly parameter, letting the data speak as to its direction. The Generalized Methods of Moments suggests that the newly added parameter is to create varying volatility fluctuation in different period discovered in financial market. Second, upward and downward jumps are separately modeled to accommodate the market data. Our model is novel and highly tractable, which means that the quasi-closed-form solutions for futures and options prices can be derived. We have employed data on VIX future and corresponding option contracts to test this model to evaluate its ability of performing pricing and capturing features of the implied volatility. To sum up, the free stochastic volatility model with asymmetric jumps is able to adequately capture implied volatility dynamics and thus it can be seen as a superior model relative to the fixed volatility model in pricing VIX derivatives.
\end{abstract}

\maketitle

\section{Introduction}\label{sec1}
Since the Chicago Board Options Exchange (CBOE) launched the CBOE Volatility Index (VIX) futures in March 2004 and later VIX options in February 2006 and a great deal of financial innovation in volatility had been traded on markets over the past few years, the trading volume of derivatives on the VIX index has grown considerably over the last decade and become popular among investors. The increasing volume of trading in those products is largely due to the fact that VIX options give investors the possibility to directly and effectively invest in volatility without having to factor in the price changes of the underlying instrument, dividends, interest rates or time to expiration. Moreover, VIX derivatives can serve as an effective hedging instrument against financial turmoil. In fact, the index is regarded as the “fear gauge” since the VIX index tends to rise when large price movement and market turmoil occur, whereas when the market is easing upward in a long-run bull market, the VIX index remains low and steady. 

\begin{figure}[htpb]
\caption{\\Plot of the VIX index against S\&P500 and VVIX(01/02/2015--1/29/2016)}
\label{Figure 1}
\subfigure[VVIX v.s. VIX]{\label{VVIX v.s. VIX}\includegraphics[width=0.31\textwidth]{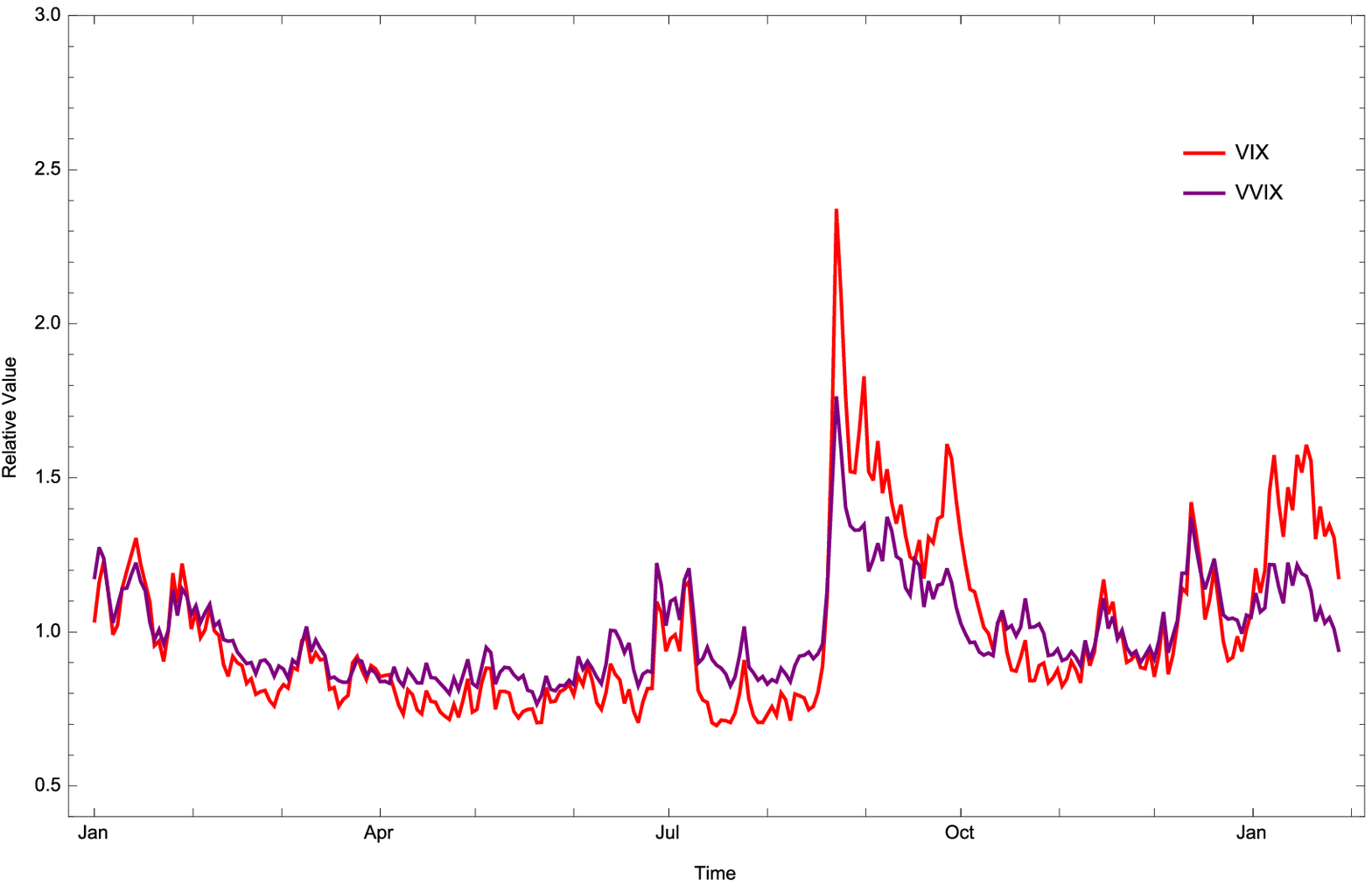}}
\subfigure[S\&P500 v.s. VIX]{\label{SPX v.s. VIX}\includegraphics[width=0.31\textwidth]{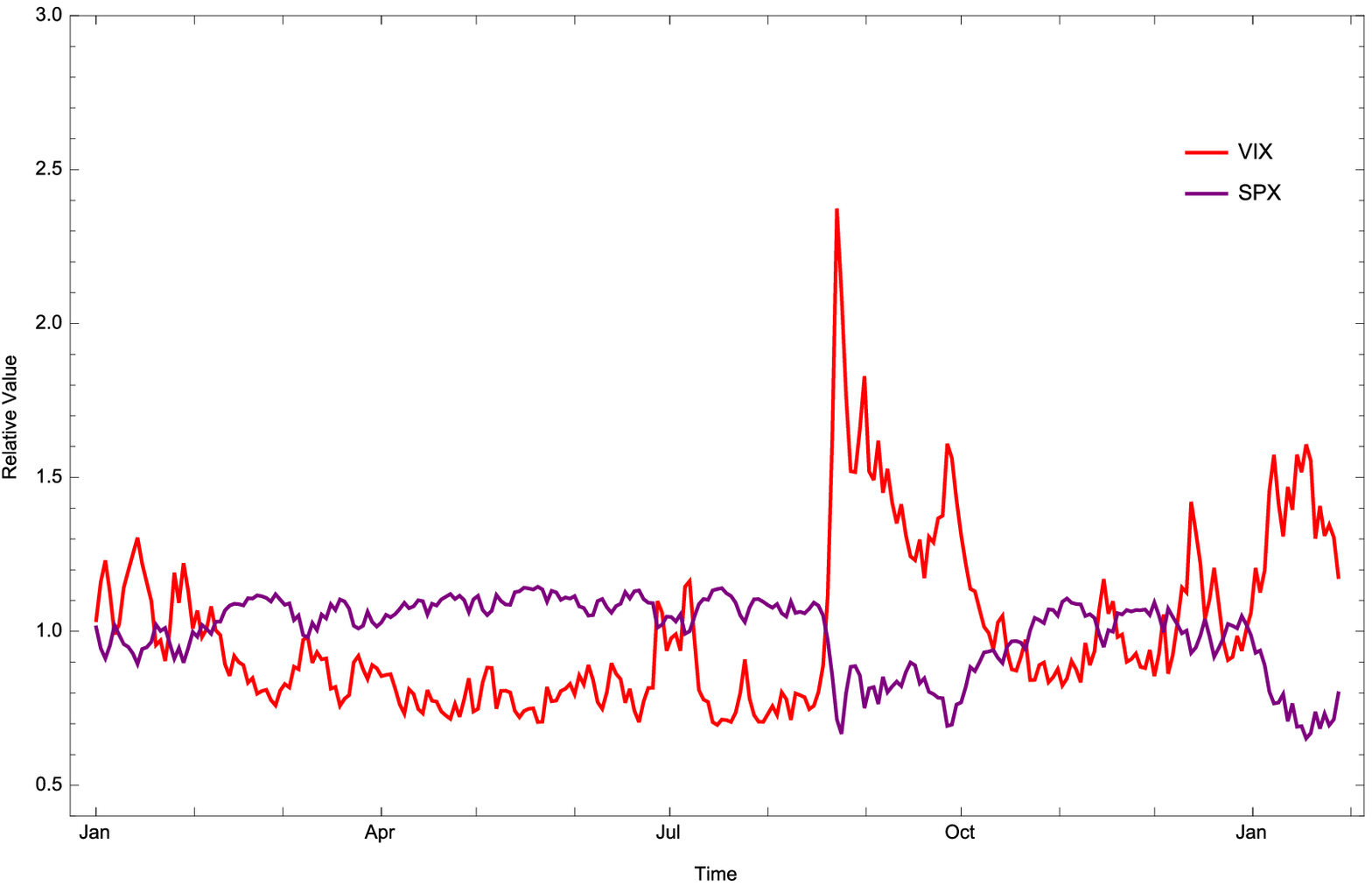}}
\subfigure[Scatter plot of VIX returns v.s. S\&P500 returns]{\label{vixspxreturn}\includegraphics[width=0.31\textwidth]{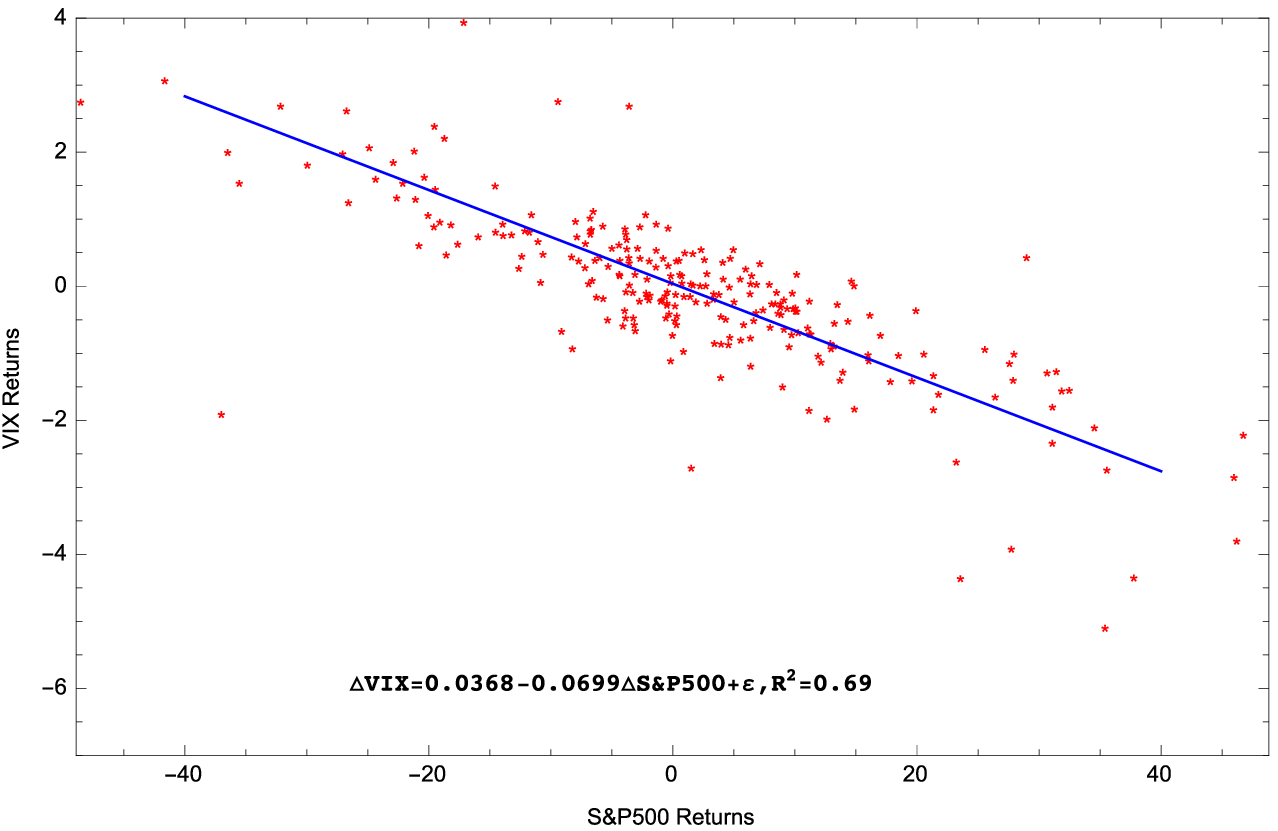}}
\end{figure}

We use the term ``\textit{free stochastic volatility}" to refer to the fact that the volatility of the VIX is not merely stochastic but also varies asymmetrically in response to changes in the S\&P500. Such asymmetry can be seen, for example, by looking at the dynamics of the VIX and the VVIX of the CBOE. Both of them not only fluctuate randomly and frequently but also change upward and downward asymmetrically. On the other hand, an interesting point of empirical regularity is that the stochastic volatility factor implicit in the VIX mainly derives from two components--one that can be spanned by S\&P500 and another that cannot. In order to gauge the extent to which VIX can be influenced by S\&P500, we run a regression of S\&P500 returns onto VIX returns, and find that the S\&P500 changes can explain 69\% of the variation in the vix based on the adjusted $\mathrm{R}^2$ of the regression. Thus, a consideration of free stochastic volatility is expected to better capture the dynamics of S\&P500 returns, which in turn better reconciles the theoretical model with VIX, and as a result, improves option valuation. Fig.\ref{Figure 1} depicts the historical dynamics of the relative VIX index along with its volatility VVIX index, S\&P500 index and scatter plot of VIX returns against S\&P500 returns respectively. Inspection of Fig.\ref{Figure 1}, a number of important features that need to be addressed if one wishes to model the S\&P500 and VIX dynamic jointly. First, there is a positive relation between changes in the VVIX and in the VIX. Second, VIX fluctuating over time shows periods of high and low volatility and VIX fluctuations tend to cluster, so that it has fueled the demand for free volatility models that are capable of flexibly reproducing the observed volatility. The third feature is the leverage effect, i.e. the S\&P500 and VIX indexes are negatively correlated. fourth, there are sign of simultaneous and oppositely directed jumps in S\&P500 and the VIX indexes.

Contrast to the empirical fact just discussed, the existing VIX option papers generally assume that the volatility of S\&P500 follows either square root or 3/2 diffusion. Based on Heston model \cite{H1993}, Grünbichler and Longstaff \cite{GL1996} priced options on instantaneous volatility, letting the volatility process follow a mean-reverting squared Bessel process (usually called CIR or square root process since it displayed a power 1/2 in the diffusion term). Later, Zhang and Zhu \cite{ZZ2006} proposed VIX futures pricing. Lian and Zhu \cite{LZ2013} extended this to the case of VIX options when the underlying index and its volatility are allowed to jump. On the other hand, the inverse of CIR process is still mean reverting and its diffusion term contains the process to the power 3/2. The related paper by Drimus \cite{D2012} studied the pricing and hedging of options on realized variance in the 3/2 non-affine model. Later, Baldeaux and Badran \cite{BB2014} derived a semi-closed formula for VIX options in a 3/2 stochastic volatility model with jumps in the underlying index. 
An interesting point to note is that Duan and Yeh \cite{DY2010} introduced a joint price-volatility model which is more general than that of Bate \cite{B2000} and Pan \cite{P2002} due to its free-root volatility process. It turns out that all of those models assume that the volatility diffusion term in equity is either square root or 3/2, i.e. a fixed power of CIR process, which make them not good enough to accurately accommodate the important feature of the data and volatility skew. Hence, we set free stochastic volatility parameter $\alpha$ to be the power of instantaneous variance instead of restricting it, letting the data speak as to its direction. We name it free stochastic volatility model (FSV). Finally, this paper aims to fill this vacuum to show the capacity of pricing VIX derivatives with FSV model.



The purpose of this paper is threefold. Firstly, we set up FSV model without jump and address a couple of technical conditions needed to avoid a well-defined FSV to explode. In this framework, the Heston and 3/2 models are only parametric example, hence one has a greater modeling freedom to control the distribution shape of asset returns time series. Since the four-parameter stochastic process, the Heston and 3/2, have empirically proven themselves quite good to provide a good control on the equity return, the additional fifth parameter $\alpha$ of the new process is expected to identify and categorize the equity distributions according to the economic behavior implicit in the prices. From this point of view, the estimation technique, Generalized Methods of Moments (GMM) of Hansen \cite{H1982}, used to empirically estimate and compare those models is outlined. This method is first established by Chan et al \cite{C1992} for comparing the short-term interest rate models. Later, Goard and Mazur \cite{GM2013} had used this method for several different models which are modelling VIX directly, without providing a connection to the underlying index. We emphasize that our estimation are conceptually different to that of them. One different part is data. Our model provides a connection to the dynamics of the underlying index, which are data for GMM estimation. The second different part is that we have done the estimation exercises for four different periods as illustrated to test whether FSV parameter $\alpha$ changes significantly with time. Finally, the conclusion that $\alpha$ of the indices implies different volatility fluctuation in different periods, and should be considered in our model.

Secondly, we establish FSV with setting upward and downward respectively. Both of them are assumed to follow independent compound Poisson processes with each having its own jump intensity and exponential jump-size distribution. Park \cite{P2016} proposed this kind of asymmetric jumps in his model. However, instead of modelling the VIX directly and adding jump in the VIX without providing underlying index, our model establish upward and downward in underlying index. Furthermore, in economically reasonable stochastic volatility models, the actual volatility should be a recurrent process without explosion. An explosion in the auxiliary volatility process causes some discounted financial claim prices to fail to be martingale; instead, they are only local martingales. For those claims, we propose and prove that the discounted stock price is a true martingale, and not just a local martingale under our FSV model.

Thirdly, we introduce four models of S\&P500 price dynamics with and without the features of free volatility and jumps. Each model can derive its quasi-analytic solutions to the prices of futures and options. The first one is the stochastic volatility (HSV) model based on Heston's model with volatility parameterized as in Cox et al. \cite{CIR1985}. The second one is the 3/2-stochastic volatility plus jumps (3/2-SVJ) model due to Baldeaux and Badran with asset returns containing price jumps. The following two models (FSV-type) relax the power parameter of instantaneous variance. The third model is the FSV model with asymmetric jumps (FSV-AJ), which means that upward and downward jumps are both assumed to follow independent compound Poisson processes with each having its own jump intensity and exponential jump-size distribution respectively. To compare FSV-AJ and investigate whether upward jumps can have an incremental effect on the pricing of VIX derivatives, the fourth model is a FSV with downward jumps only model (FSV-DJ). These consistent models contain information about VIX indexes, VIX futures and VIX options. Then, models are tested on the foregoing data covering span between March 1, 2016 and March 31, 2016.

Our paper investigate the effects of FSV parameter $\alpha$ on the pricing of VIX derivatives. Comparing the FSV-type with 3/2-SVJ and HSV, we find that former is strongly and obviously preferred to the latter in both in-sample and out-of-sample tests. By comparing FSV-DJ and FSV-AJ model, the in-sample and out-of-sample performance of FSV-AJ show lightly better than FSV-DJ, especially in OTM options. We plot the VIX future values as a function of time to maturities under four models, and presented in Fig.~\ref{Figure 3}. In addition, we also plot those of the VIX option values as a function of strikes and presented in Fig.~\ref{Figure 10}. That is, allowing for FSV parameter $\alpha$ in model make large improvements in fitting the prices of VIX futures and options. We next look at the effects of setting upward jumps and downward jumps respectively. Good and bad surprises may arrive with different rates and sizes and investors may react differently to them. Hence, our models assume that upward and downward jumps occur independently with different frequencies and magnitudes. Results show that including upward jump contribute to pricing OTM options. Finally, we find that a joint consideration of the FSV parameter $\alpha$ and asymmetric jumps is expected to better capture the dynamics of equity returns, which in turn better reconciles the theoretical model with the observed volatility smile/smirk, and as a result, improves valuation of VIX futures and options.

The balance of this paper is organized as follows. In Section 2, we propose FSV model under risk-neutral measure and its non-explosion condition. The estimation technique, GMM, used to empirical test is outlined. Section 3, we present the general FSV-AJ model setup and prove the discounted stock price is a martingale. Section 4, pricing formulae of VIX future and option are provided. The VIX options and futures data are described in Section 5. Section 6 shows the parameters estimates and some preliminary analysis, while Section 7 provides the main empirical result on pricing performance across the different model specification. Concluding remarks are offered in Section 8.

\section{Free stochastic volatility model Applied to the S\&P500}
In this Section, we introduce free stochastic volatility model and provide numerical result to illustrate that this kind of model is capable of producing implied volatility skews in VIX options. Consider a risk-neutral probability space ($\varOmega$, $\mathscr{F}$, $\mathbb{Q}$) and information filration $\{\mathscr{F}_t\}$, where the price process $S_t$ is adapted to the filtration $\{\mathscr{F}_t\}_{t\ge0}$. The general dynamics for stock price and variance process include free stochastic volatility parameter $\alpha$ and Cox-Ingersoll-Ross (CIR) process:

\begin{equation}\label{eq0.1}
dS_t=S_t\left(rdt+\gamma V_t^\alpha dW_t^{\mathbb{Q}}\right)
\end{equation}
where the stochastic factor $V$ evolves as
\begin{equation}\label{eq0.2}
dV_t=\kappa \left(\theta-V_t\right) dt+\sigma \sqrt{V_t}dZ_t^{\mathbb{Q}},
\end{equation}
starting at $S_0>0$ and $V_0>0$, respectively. As per usual, $r,\kappa, \theta$ are assumed to be strictly positive. $\theta$ controls for the long-term mean of $V_t$ and $\kappa$ is the mean-reversion speed of $V_t$. Furthermore, $\sigma\in\mathbb{R}$ captures the volatility of $V_t$. $W_t^{\mathbb{Q}}$ and $Z_t^{\mathbb{Q}}$ are standard Brownian motions under risk-neutral $\mathbb{Q}$ measure. $\alpha\in\lbrack-\frac{1}{2},\frac{3}{2}\rbrack$ stands for free stochastic volatility part so that we wish it can capture more complicated fluctuation volatility implied in S\&P500. To capture the market behavior of S\&P500 and VIX index, we hope for a negative correlation, denoted by $-1\le\rho<0$, between $dZ_t^{\mathbb{Q}}$ and $dW_t^{\mathbb{Q}}$. Our joint price-volatility model is more general than that of Heston model and 3/2 stochastic models because their specifications correspond to the special case of $\alpha$ to be fixed.

We now address a couple of technical condition needed to have a well-define free stochastic volatility model. In order to avoid $V_t^{\alpha}$ blowing up, some technical conditions with parameters should be satisfied. First, apply the It$\hat{\text{o}}$ formula to $D_t=V_t^{\alpha}$, we get
$$dD_t=\left(\alpha D_t^{1-\frac{1}{\alpha}}\left(\kappa\theta+(\alpha-1)\frac{\sigma^2}{2}\right)-\alpha\kappa D_t\right)dt+\alpha\sigma D_t^{1-\frac{1}{2\alpha}}dZ_t.$$
Its scale density is
$$s(D)=D^{-\frac{2\tilde{\theta}}{\alpha\sigma^2}}e^{\frac{2\kappa}{\sigma^2}D^{\frac{1}{\alpha}}}$$, using $\tilde{\theta}=\kappa\theta+(\alpha-1)\frac{\sigma^2}{2}$. Hence, the scale measure $S(c,d)$ becomes
$$S(c,d)=\int_c^{d}s(x)dx=\int_c^{d}x^{-\frac{2\tilde{\theta}}{\alpha\sigma^2}}e^{\frac{2\kappa}{\sigma^2}x^{\frac{1}{\alpha}}}dx.$$
In particular, we see that
\begin{equation}
S(c,+\infty)=\left\{
\begin{array}{l}
\infty\quad\quad\text{for}\hspace{0.1cm}(\text{i})\hspace{0.1cm}\alpha>0\hspace{0.1cm}\text{or}\hspace{0.1cm}(\text{ii})\hspace{0.1cm}\alpha<0\hspace{0.1cm}\text{and}\hspace{0.1cm}\tilde{\theta}>0\\
<\infty\quad\text{for}\hspace{0.1cm}\alpha<0\hspace{0.1cm}\text{and}\hspace{0.1cm}\frac{2\tilde{\theta}}{\alpha\sigma^2}>1\\
\end{array}
\right\}.
\end{equation}
Hence, if $\tilde{\theta}>0$, $S(c,+\infty)$ can totally be divergent no matter what $\alpha$ it is. $D=\infty$ is classified as an natural or entrance boundary under boundary classification criteria. Fortunately, both kinds of boundaries make $D=\infty$ unreachable from the interior in finite time, which means $D=V_t^{\alpha}$ never explode in finite time.

On the other hand, CIR process is a space-time change BESQ process. Using Feller's boundary condition, $V_t$ process starting from a positive initial point stays strictly positive and never explodes in finite interval, if and only if $2\kappa\theta\ge\sigma^2.$ Hence, it will also reject $S(c,+\infty)<\infty$ due to the equivalence of $2\kappa\theta\ge\sigma^2$ and $\frac{2\tilde{\theta}}{\alpha\sigma^2}\le1$. Collecting our result, we have establish $\tilde{\theta}>0$, i.e.
$$\frac{2\kappa\theta}{\sigma^2}>1-\alpha$$
to be technical condition of parameters.


In order to provide the readers with parameters that can be testified. Following Goard and Mazur \cite{GM2013}, among others, we use the same estimation technique to estimate the parameters in the continuous-time model and create hypothesis testing to see whether these parameters impose unreasonable overidentifying restrictions. We emphasize that our estimation are conceptually different to that of Goard and Mazur \cite{GM2013}. One different part is data. Instead of modelling the VIX directly, without providing a connection to the underlying index, our model for VIX in \ref{eq0.1} and \ref{eq0.2} is derived directly from the dynamics of the underlying index. As the consequence, we use S\&P500 data to estimate, not VIX data. On the other hand, we use the most recent one-year period S\&P500 index data: January 2, 2016--December 2, 2016. We have done the estimation exercises for four different periods: the whole period from January 2, 2016--December 2, 2016 as testified whether these parameters impose unreasonable overidentifying restrictions on each model, January 2, 2016--April 2, 2016, May 2, 2016--August 2, 2016 and September 2, 2016--December 2, 2016 as illustrated to test whether our free stochastic volatility parameter $\alpha$ changes significantly with different periods in 2016. Referred to \autoref{app}  for detailed computation process, we rewrite the corresponding discrete-time econometric specification:

\begin{align}\label{eq2.41}
\frac{S_{t+1}-S_{t}}{S_{t}}&=r\Delta t+\gamma\varepsilon_{t+1}\\
\mathbb{E}^{\mathbb{Q}}\left[\varepsilon_{t+1}\right]&=0
\end{align}
\begin{align}\label{eq2.51}
\mathbb{E}^{\mathbb{Q}}\left[\varepsilon_{t+1}^{2}\right]&=\frac{\Gamma\left[2\left(\alpha+\frac{\kappa\theta}{\sigma^2}\right)\right]}{\Gamma{\left[\frac{2\kappa\theta}{\sigma^2}\right]}}\left(\frac{\sigma^4}{4}\right)^{\alpha}e^{-2\kappa\left(t\left(\alpha+\frac{\kappa\theta}{\sigma^2}\right)+\frac{V_0}{\left(e^{t\kappa}-1\right)\sigma^2}\right)}\notag\\
&\quad\times\left(\frac{\kappa}{1-e^{-t\kappa}}\right)^{\frac{2\kappa\theta}{\sigma^2}}\left(\frac{\kappa}{-1+e^{t\kappa}}\right)^{-2\alpha-\frac{2\kappa\theta}{\sigma^2}} \ _{1}F_1\left[2\left(\alpha+\frac{\theta\kappa}{\sigma^2}\right),\frac{2\kappa\theta}{\sigma^2},\frac{2V_{0}\kappa}{\left(-1+e^{t\kappa}\right)\sigma^2}\right]\Delta t \\
&\triangleq \phi(\kappa, \theta, \sigma, \alpha)\Delta t \notag
\end{align}
where $\ _{1}F_1$ is the Kummer confluent hypergeometric function.

Our econometric approach is to estimate parameters of this process and test \eqref{eq2.41}-\eqref{eq2.51} as a set of overidentifying restrictions on a system of moment equations using the GMM of Hansen \cite{H1982}. Advantages of the method, as state by Chan et al \cite{C1992} make it an intuitive and logical choice for the estimation of the continuous-time volatility processes. First, it makes no assumptions about the distributional nature of the changes in S\&P500. And the asymptotic justification for the GMM procedure requires only that the distribution changes be stationary and ergodic and that the relevant expectations exist. As well GMM estimators and their standard errors are consistent even if the disturbances are conditionally heteroskedastic.

Define $m_{t}(\omega)\phi$ be parameter vector with elements $\kappa, \theta, \sigma, \gamma$ and $\alpha$ and let the vector $m_{t}(\omega)$ be
\begin{equation}
m_{t}(\omega)=\left[
\begin{array}{c}
\varepsilon_{t+1}\\
\varepsilon_{t+1}S_{t}\\
\varepsilon_{t+1}S_{t}^2\\
\varepsilon^2_{t+1}-\phi(\kappa, \theta, \sigma, \alpha)\Delta t\\
\left(\varepsilon^2_{t+1}-\phi(\kappa, \theta, \sigma, \alpha)\Delta t\right)S_{t}\\
\end{array}
\right]
\end{equation}

Under the null hypothesis that the restrictions implied in equations \eqref{eq2.41}-\eqref{eq2.51} are true, the orthogonality conditions, $\mathbb{E}^{\mathbb{Q}}[m_{t}(\omega)]=0$. The GMM technique consists of replacing $\mathbb{E}^{\mathbb{Q}}[m_{t}(\omega)]$ with its sample counterpart, $M_{T}(\omega)$, using the T observations where $M_{T}(\omega)=\frac{1}{T}\sum^{T}_{t=1}m_t(\omega)$, and then choosing parameters which minimize the quadratic form 
\begin{equation}\label{eq2.81}
J_{T}(\omega)=M_{T}'(\omega)W_{T}(\omega)M_{T}(\omega),
\end{equation}
where $W_{T}(\omega)$ is a positive definite symmetric weighting matrix. The GMM estimates of the overidentified parameter subvector of $\omega$ do depend on the choice of $W_{T}(\omega)$. Hansen \cite{H1982} provided that setting $W_{T}(\omega)=S^{-1}(\omega)$, where $S(\omega)=\mathbb{E}^{\mathbb{Q}}[m_{t}(\omega)m'_{t}(\omega)]$, delivers the GMM estimator of $\omega$ with the smallest asymptotic covariance matrix. The minimized value of the quadratic form in \eqref{eq2.81} is distributed $\chi^{2}$ under the null hypothesis that the model is true with degrees of freedom equal to the number of orthogonality conditions net of the number of parameters to be estimated. This $\chi^{2}$ measure provides a goodness-of-fit test for the model. A hypothesis test is then used to test whether the models impose unreasonable overidentifying restrictions upon the unrestricted model, i.e., for each nested model, we create the hypothesis test of $a_{0}$ versus $a_{1}$ where

\quad\quad $a_{0}$: The model does not impose overidentifying restrictions and is hence not misspecified

\quad\quad $a_{1}$: The model does impose overidentifying restrictions and is hence misspecified

The test statistic, $R=T\left[J_{T}(\tilde{\omega})-J_{T}(\hat{\omega})\right]$, is asymptotically distributed $\chi^2$ with degrees of freedom equal to the number of restrictions on the general model to obtain the nested model. By using the same weighting matrix from the unrestricted model, this test statistic is the normalized difference of the restricted $J_{T}(\tilde{\omega})$ and unrestricted $J_{T}(\hat{\omega})$ objective functions for the efficient GMM estimator. A high value of this statistic means that the model is misspecified. If the p--value is less than the required level of significance then we can draw a conclusion that this model is misspecified.

With a view to a fully and complete description, the GMM results are presented in Table \ref{Table -1}. Parameters of all estimations have small standard deviations and hence are stable. First, regarding the reported $\chi^2$ values estimated in the whole period, Heston and 3/2 models are rejected at the 1\% level of significance, be quite informative to explain the internal working of each model to fit S\&P500 data. Hence, there models are misspecified and place irrational restrictions on the unrestricted model. On the other hand, however, free stochastic model results in acceptance at the 1\% level significance with p-value of 0.767, due to the reason that it contains a new parameter $\alpha$ contributing to fit S\&P500. Thus, the model is not misspecified. Third, we also estimate the parameters for the free stochastic model with different periods in 2016. Given that free stochastic model has $\alpha=0.864$ in period A, $\alpha=0.901$ in period B, and $\alpha=0.943$ in period C, we'd like to draw the conclusion that $\alpha$ of the indices implies different volatility fluctuation in different periods, and should be considered in our model. Heuristically, the free stochastic model with free volatility parameter $\alpha$ fit S\&P500 that is widely regarded as frequent fluctuation index quite well while both Heston and 3/2 model are rejected.

\begin{table}
\renewcommand\arraystretch{1.2}
\caption{\\We estimate the parameters for different models processes nested within equation \eqref{eq0.1} with its standard error in parentheses to test the significance of the individual parameters. The $\chi^2$ test statistics are computed following the method outlined in Newey and West \cite{W1985} with p-value in parentheses and associated degrees of freedom (DF). GMM method estimates the values of parameters from the historical data over four different periods: The whole period, period A, period B, period C. Period A--C are only used in free stochastic model.}
\label{Table -1}
\centering
\begin{tabular*}{\textwidth}{@{\extracolsep{\fill}}llllllllll}
\hline
\multicolumn{3}{l}{Model}  & $\kappa$ & $\theta$ & $\sigma$ & $\gamma$  & $\alpha$ & $\chi^2$ & DF \\ 
\hline
\leftline{The whole period: January 2, 2016--December 2, 2016}\\
\multicolumn{3}{l}{Unrestricted} & 1.513 & 0.475 & 0.509 & 1.011 & 1.274 &  N/A  & N/A \\
\multicolumn{3}{l}{}            &   $<[0.001]$ & $<[0.001]$ &$<[0.001]$ & $<[0.012]$   &       $<[0.08]$   & & \\
\multicolumn{3}{l}{Heston} & 2.829 & 0.020 & 0.831 & 1 & 0.5 &16.583 & 2\\
\multicolumn{3}{l}{}        & $<[0.001]$    & $<[0.001]$ &$<[0.001]$   &    &                      &$<[0.001]$& \\
\multicolumn{3}{l}{3/2 model}& 10.29 & 56.49 & -2.598 & 1 & -0.5 &15.457 & 2\\
\multicolumn{3}{l}{}            & $<[0.001]$ &  $<[0.06]$ &  $<[0.03]$       &    &                      &$<[0.001]$& \\
\multicolumn{3}{l}{Free stochastic}& 2.036 & 0.502 & 0.650 & 1 & 1.288  &0.087 & 1\\
\multicolumn{3}{l}{}     &    $<[0.001]$     &    $<[0.001]$       &    $<[0.001]$    &    &      $<[0.01]$     & 0.767& \\
\hline
\leftline{Period A: January 2, 2016--April 2, 2016}\\
\multicolumn{3}{l}{Free stochastic}& 1.699 & 0.617 & 0.484 & 1 & 0.864  &3.866 & 1\\
\multicolumn{3}{l}{} &   $<[0.001]$  & $<[0.001]$     &    $<[0.001]$    &  &    $<[0.02]$         & 0.0492& \\
\hline
\leftline{Period B: May 2, 2016--August 2, 2016}\\
\multicolumn{3}{l}{Free stochastic}& 2.326 & 0.695 & 0.687 & 1 & 0.901  &1.654 & 1\\
\multicolumn{3}{l}{}   &   $<[0.001]$      &     $<[0.001]$      &    $<[0.001]$    &    & $<[0.07]$            & 0.1983& \\
\hline
\leftline{Period C: September 2, 2016--December 2, 2016}\\
\multicolumn{3}{l}{Free stochastic}& 1.911 & 0.387 & 0.469 & 1 & 0.943  &1.119 & 1\\
\multicolumn{3}{l}{}  & $<[0.001]$        &     $<[0.001]$      &$<[0.001]$        &    &           $<[0.05]$& 0.289& \\
\hline
\end{tabular*}
\end{table}

\section{Free stochastic volatility plus asymmetric jumps model}\label{sec2}
As stated in Sepp \cite{S2008}. The author asserts that `stochastic volatility models without jumps are not consistent with the implied volatility skew observed in options on the VIX...' and that `...only the stochastic volatility with appropriately chosen jumps can fit the implied VIX skew'. However, the jump sizes modeled in Heston \cite{H1993}, Baldeaux and Badran\cite{BB2014}, even Lian and Zhu \cite{LZ2013} are distributed to be normal in underlying dynamics. Normal jump makes upward and downward jumps symmetric due to its central mean and variance. Furthermore, a closely related notion is the observation that the equity market is often more volatile on the downside than the upside. Because of that, in this section, upward and downward jumps in underlying dynamics are assumed respectively. Both are driven by independent compound Possion processes. We first extends free stochastic volatility to the asset-price process with asymmetric jumps. That's a new family of free volatility and jump-diffusion models for the dynamics of VIX. We derive semi-closed-form solution to the prices of futures and options. In addition, we provide a description of our estimation method and introduce a competing model from 3/2 model plus jumps and Heston model. The Heston model is a limiting case of every other model and viewed as the benchmark in our empirical study. We then demonstrate that the VIX formula is still valid when the underlying asset price process includes FSV parameter $\alpha$ and asymmetric jumps.

Consider the dynamics for the underlying index given by:
\begin{equation}\label{eq2.1}
\frac{dS_t}{S_t}=(r-\lambda_1\tilde{\mu}_1-\lambda_2\tilde{\mu}_2)dt+V_t^\alpha dW_t^{\mathbb{Q}}+(e^{J_1^{\mathbb{Q}}}-1)dN_{1t}^{\mathbb{Q}}+(e^{J_2^{\mathbb{Q}}}-1)dN_{2t}^{\mathbb{Q}}
\end{equation}
where the stochastic factor $V$ evolves as
\begin{equation}\label{eq2.2}
dV_t=\kappa \left(\theta-V_t\right) dt+\sigma \sqrt{V_t}dZ_t^{\mathbb{Q}}.
\end{equation}
As for asymmetric jumps,  we again emphasize that this jump is conceptually different to that of Park \cite{P2016}. Instead of modelling the VIX directly and adding jump in the VIX without providing underlying index, we establish upward and downward in underlying index. Specifically, we assume that they are driven by independent compound Possion processes, with each having its own jump intensity and jump-size distribution. $N_{1t}^{\mathbb{Q}}$ and $N_{2t}^{\mathbb{Q}}$ denote risk-neutral Possion processes driving upward and downward jumps with jump intensities  $\lambda_1$ and $\lambda_2$, respectively. Upward jump magnitudes, $J_1^{\mathbb{Q}}$, are assumed to follow an independent exponential distribution with a positive mean, $\mu_1>0$, with the probability density function taking $\frac{1}{\mu_1}\exp\{-x/\mu_1\}$ if $x>0$ and 0 otherwise where the paramters $\mu_1, \tilde{\mu}_1$ satisfy the following relationship: $$\tilde{\mu}_1=\frac{1}{1-\mu_1}-1$$. With that assumption, $e^{J_1^{\mathbb{Q}}}-1>0$ satisfies upward jump condition. The term $\lambda_1\tilde{\mu}_1dt$ is used to center the Possion innovation so that $(e^{J_1^{\mathbb{Q}}}-1)dN_{1t}^{\mathbb{Q}}-\lambda_1\tilde{\mu}_1dt$ has its mean equal to 0. Similarly, downward jump magnitudes,  $J_2^{\mathbb{Q}}$, are assumed to follow an independent exponential distribution with a negative mean, $\mu_2<0$, with the probability density function taking $\frac{1}{|\mu_2|}\exp\{-x/\mu_2\}$ if $x<0$ and 0 otherwise where the paramters $\mu_2, \tilde{\mu}_2$ satisfy the following relationship: $$\tilde{\mu}_2=\frac{1}{1-\mu_2}-1$$. With that assumption, $-1<e^{J_2^{\mathbb{Q}}}-1<0$ reaches downward jump condition. The term $\lambda_2\tilde{\mu}_2dt$ is used to center the Possion innovation so that $(e^{J_2^{\mathbb{Q}}}-1)dN_{2t}^{\mathbb{Q}}-\lambda_2\tilde{\mu}_2dt$ has its mean equal to 0.

Intergrating \eqref{eq2.1} yields
\begin{equation}\label{eq2.3}
S_t=\tilde{S}_t\prod_{s=1}^{N_{1t}}e^{J_{1s}^\mathbb{Q}}\prod_{s=1}^{N_{2t}}e^{J_{2s}^{\mathbb{Q}}}
\end{equation}
where
\begin{equation*}
\tilde{S}_t=S_0\exp\left\{\int_0^t\left(r-\lambda_1\tilde{\mu}_1-\lambda_2\tilde{\mu}_2-\frac{1}{2}V_s^{2\alpha}\right)ds+\int_0^tV_s^{\alpha}dW_s^{\mathbb{Q}}\right\}
\end{equation*}
and $J_{1s}^\mathbb{Q}$ and $J_{2s}^\mathbb{Q}$ denote the logarithm of the relative positive jump and negative jump size of $s$th jump, respectively. Since the model in Equation \ref{eq2.1} and \ref{eq2.2} is not affine, Equation \ref{eq2.3} gives us an important starting point for our analysis. In particular, one can now determine that the discounted stock price is a true martingale, and not just a local martingale under our assumed model.

\begin{prop}
Let S and V be given by Equations \ref{eq2.1} and \ref{eq2.2} respectively. Then the discounted stock price $\bar{S}_t=\frac{S_t}{e^{rt}}$ is a true martingale, and not just a strict local martingale under $\mathbb{Q}$, if and only if:
\begin{align*}
\frac{2\kappa\theta}{\sigma^2}\ge1
\end{align*}
\begin{proof}
We compute
\begin{align}\label{eq2.5}
\mathbb{E}^{\mathbb{Q}}\left[\bar{S}_T|\mathscr{F}_t\right]&=\mathbb{E}_{t}^{\mathbb{Q}}\left[\frac{\tilde{S}_T\prod_{s=1}^{N_{1T}}e^{J_{1s}^\mathbb{Q}}\prod_{s=1}^{N_{2T}}e^{J_{2s}^{\mathbb{Q}}}}{e^{rT}}\right]\notag\\
&=\mathbb{E}_{t}^{\mathbb{Q}}\left[\frac{\tilde{S}_t\prod_{s=1}^{N_{1T}}e^{J_{1s}^\mathbb{Q}}\prod_{s=1}^{N_{2T}}e^{J_{2s}^{\mathbb{Q}}}}{e^{rt}}\cdot\frac{\exp\left\{\int_t^{T}\left(r-\lambda_1\tilde{\mu}_1-\lambda_2\tilde{\mu}_2-\frac{1}{2}V_s^{2\alpha}\right)ds+\int_t^TV_s^{\alpha}dW_s^{\mathbb{Q}}\right\}}{e^{r(T-t)}}\right]\notag\\
&=\bar{S}_t\cdot\mathbb{E}_{t}^{\mathbb{Q}}\left[\exp\left\{\int_t^{T}\left(-\lambda_1\tilde{\mu}_1-\lambda_2\tilde{\mu}_2-\frac{1}{2}V_s^{2\alpha}\right)ds+\int_t^TV_s^{\alpha}dW_s^{\mathbb{Q}}\right\}\prod_{s=N_t+1}^{N_{1T}}e^{J_{1s}^\mathbb{Q}}\prod_{s=N_t+1}^{N_{2T}}e^{J_{2s}^{\mathbb{Q}}}\right]\notag\\
&=\bar{S}_t\cdot\mathbb{E}_{t}^{\mathbb{Q}}\left[\exp\left\{\int_t^{T}\left(-\lambda_1\tilde{\mu}_1-\lambda_2\tilde{\mu}_2-\frac{1}{2}V_s^{2\alpha}\right)ds+\int_t^TV_s^{\alpha}dW_s^{\mathbb{Q}}\right\}\right]e^{\lambda_1(T-t)\tilde{\mu}_1+\lambda_2(T-t)\tilde{\mu}_2}\notag\\
&=\bar{S}_t\cdot\mathbb{E}_{t}^{\mathbb{Q}}\left[-\frac{1}{2}\int_t^{T}V_s^{2\alpha}ds+\int_t^TV_s^{\alpha}dW_s^{\mathbb{Q}}\right]\notag\\
&=\bar{S}_t\cdot\mathbb{E}_{t}^{\mathbb{Q}}\left[-\frac{\rho^2}{2}\int_t^{T}V_s^{2\alpha}ds+\rho\int_t^TV_s^{\alpha}dZ_s^{\mathbb{Q}}\right]\notag\\
&=\bar{S}_t\cdot\mathbb{E}_{t}^{\mathbb{Q}}\left[\xi_{t,T}\right]
\end{align}
where we define the process $\xi=\{\xi_t, t\ge0\}$ via $$\xi_t:=\exp\left\{-\frac{\rho^2}{2}\int_0^{t}V_s^{2\alpha}ds+\rho\int_0^tV_s^{\alpha}dZ_s^{\mathbb{Q}}\right\}.$$ Obviously, it is exponential local martingale.
In order to see whether the process $\bar{S}_t=\frac{S_t}{e^{rt}}$ is a martingale, the Feller non-explosion test for $V_t$ must be satisfied under both historical and risk neutral probability measures. First, under the risk neutral probability measure $\mathbb{Q}$ the process $V_t$ cannot explode to $\infty$ and does not reach 0 if Feller condition is satisfied, i.e.
$$2\kappa\theta\ge\sigma^2$$
On the other hand, following Lewis \cite{L2000}, it involves change of Brownian motion for the volatility process: 
$$d\hat{W}_t=dW_t-\rho V_t^{\alpha}dt,$$
where $d\hat{W}_t$ is a Wiener process under the historical measure. Under this measure the auxiliary volatility process $V_t$ solves
$$d\hat{V}_t=\left(\kappa\theta-\kappa V_t+\sigma\rho V_t^{\alpha+\frac{1}{2}}\right)dt+\sigma \sqrt{V_t}d\hat{W}_t.$$
We apply the Feller explosion test which was explained in Sec. 3 of Lewis \cite{L2000}. The scale density 
$$s(V)=V^{-\frac{2\kappa\theta}{\sigma^2}}e^{\frac{2\kappa}{\sigma^2}V-\frac{2\rho}{\sigma}\frac{1}{\alpha+\frac{1}{2}}V^{\alpha+\frac{1}{2}}}$$ and so the scale measure is given by
$$S(c,d)=\int_c^d s(V)dV=\int_c^dV^{-\frac{2\kappa\theta}{\sigma^2}}e^{\frac{2\kappa}{\sigma^2}V-\frac{2\rho}{\sigma}\frac{1}{\alpha+\frac{1}{2}}V^{\alpha+\frac{1}{2}}}dV\quad\text{for} \quad0<c<d.$$
Assuming $\rho<0$, then $S(c,+\infty)=\infty$ if and only of $\alpha+\frac{1}{2}>0$. Since $\alpha\in[-\frac{1}{2},\frac{3}{2}]$ is asked in our model, the divergence result follows. For the speed density, we have
$$m(V)=\frac{1}{\sigma^2Vs(V)}=\frac{1}{\sigma^2}V^{\frac{2\kappa\theta}{\sigma^2}-1}e^{-\frac{2\kappa}{\sigma^2}V+\frac{2\rho}{\sigma}\frac{1}{\alpha+\frac{1}{2}}V^{\alpha+\frac{1}{2}}}.$$
Clearly, $N(\infty)=\lim\limits_{d\uparrow\infty}\int_c^dS(c,x)m(x)dx$ diverges as $d\to\infty$.  According to boundary classification criteria, it show that $S(c,+\infty)=\infty$ and $N(\infty)=\infty$ suffice to classify $V=\infty$ as an natural boundary under a negative correlation coefficient, which means there is no explosion in this case, since the boundary is unreachable in finite time. Collecting our result, the process $\bar{S}_t$ is a martingale. 
\end{proof}
\end{prop}

\section{Pricing of VIX options and futures}
In this section, we derive a general pricing formula for European call options and futures on the VIX when the index follows FSV process. First, the log contract can then be synthesized with a portfolio of call and put options in a continuum of strikes (Breeden Litzenberger) \cite{BL1978}, which leads to the VIX formula (CBOE, 2003) \cite{CBOE2003}. Thus, VIX squared can be expressed in terms of the risk-neutral expectation of the log contract. Different dynamics for the index price $S_t$ will result in various expressions for VIX squared. Hence, the squared VIX index is an approximation to the value of the log contract: 
\begin{equation}\label{eq3.1}
\mathrm{VIX}_t^2\approx-\frac{2}{\tau}\mathbb{E}^\mathbb{Q}\left[{\log \left({\frac{S_{t+\tau}}{S_te^{r\tau}}}\right)\bigg|\mathcal{F}_t}\right]\times100^2,
\end{equation}
with $\tau=\frac{30}{365}$ and $S_te^{r\tau}$ being forward price of S\&P500 observed at time $t$ with $t+\tau$ as maturity. The VIX formula involves two approximation errors:

\quad(1) an error if the price dynamics includes jumps: and

\quad(2) if options are only available for a finite number of strikes. 

In the following subsections, we go though Heston and 3/2 model briefly one after another and present the VIX derivatives pricing formulae in FSV model. The following Theorem \ref{theo3.2} derivation of VIX options pricing formula, which is an extension of Proposition 3.4 in Baldeaux and Badran \cite{BB2014}, also extends Proposition 1 in Zhang and Zhu \cite{ZZ2006}. 


\subsection{Heston stochastic volatility (HSV)}
Setting the parameter of volatility equals to $\frac{1}{2}$ and $\lambda_1=\lambda_2=0$ in equation \eqref{eq2.1}, we obtain the Heston model. The Heston model was first proposed by Heston and it has been extensively used and studies due to its tractability. In this specification of HSV, Lian
and Zhu \cite{LZ2013} show that the expectation in \eqref{eq3.1} can be computed explicitly: $\mathrm{VIX}_t^2=100^2\times\left(aV_t+b\right)$ where $a=\frac{1-e^{\kappa\tau}}{\kappa\tau}$ and $b=\theta\left[1-\frac{1-e^{\kappa\tau}}{\kappa\tau}\right].$ With the transitional probability density function (TPDF) $f^{\mathbb{Q}}_{V_T | V_t}({y})$ of CIR process, inversion of $\mathrm{VIX}_t^2=100^2\times\left(aV_t+b\right)$ gives us the TPDF of the VIX index $f^{\mathbb{Q}}_{\mathrm{VIX}_T | \mathrm{VIX}_t}(z)=\frac{2z}{a100^2}\cdot f^{\mathbb{Q}}_{V_T | V_t}\left(\frac{\frac{z^2}{100^2}-b}{a}\right)\mathds{1}_{\left\{z\ge100\sqrt{b}\right\}}.$ Hence, the price of a European call option is then found by computing the expected payoff directly as:
\begin{equation}
\text{C}\left(\mathrm{VIX}_t,K,t,T\right)=e^{-r\left(T-t\right)}\int_K^{\infty}\max\left(y-K,0\right)f^{\mathbb{Q}}_{\mathrm{VIX}_T | \mathrm{VIX}_t}(y)dy,
\end{equation}
while the VIX futures equal:
\begin{equation}
\text{F}\left(\mathrm{VIX}_t,t,T\right)=\mathbb{E}^{\mathbb{Q}}\left[\mathrm{VIX}_T | \mathcal{F}_t\right]=\int_0^{\infty}y\cdot f^{\mathbb{Q}}_{\mathrm{VIX}_T | \mathrm{VIX}_t}(y)dy.
\end{equation}

Before pricing VIX derivatives in FSV and 3/2 models, we need to establish the following Lemma \ref{lem3.2} first.

\begin{lemma}\label{lem3.2}
Let $X^x=\left\{ X_t^x,t \ge 0  \right\}$ denote the solution of the \eqref{eq2.2} (CIR) SDE and $X_0=x >0$ with $\kappa,\theta,\sigma >0$ and $2\kappa\theta \ge \sigma^2$ (Feller condition). Consider $\epsilon, \nu, \eta, \gamma \in \mathbb{R}$ such that
\begin{align}
\epsilon&>-\frac{\kappa^2}{2\sigma^2},\\
\nu&\ge -\frac{\left({\kappa\theta-\frac{\sigma^2}{2}}\right)^2}{2\sigma^2},\\
\eta&<\frac{\kappa\theta+\frac{\sigma^2}{2}+\sqrt{\left({\kappa\theta-\frac{\sigma^2}{2}}\right)^2+2\sigma^2\nu}}{\sigma^2},\\
\gamma&\ge-\frac{\sqrt{\kappa^2+2\epsilon\sigma^2}+\kappa}{\sigma^2}.
\end{align}
The following transform for the CIR process is well defined for all $t\ge0$ and is given by 
\begin{align}
\phi(t,x;\eta,\gamma,\epsilon,\nu)&=\mathbb{E}\left[ ({X_t^x})^{-\eta}\exp\left({-\gamma X_t^x-\epsilon \int_0^t X_t^x ds-\nu\int_0^t\frac{ds}{X_t^x}}\right)\right]\notag\\
&=\left({\frac{\beta({t,x})}{2}}\right)^{m+1}x^{-\frac{\kappa\theta}{\sigma^2}}({\gamma+K({t})})^{-\left({\frac{1}{2}+\frac{m}{2}-\eta+\frac{\kappa\theta}{\sigma^2}}\right)}\notag\\
&\quad\times e^{\frac{1}{\sigma^2}\left({\kappa^2\theta t-\sqrt{A}x\coth\left({\frac{\sqrt{A}t}{2}}\right)+\kappa x}\right)}\frac{\Gamma\left({\frac{1}{2}+\frac{m}{2}-\eta+\frac{\kappa\theta}{\sigma^2}}\right)}{\Gamma({m+1})}\notag\\
&\quad\times _{1}F_{1}\left({\frac{1}{2}+\frac{m}{2}-\eta+\frac{\kappa\theta}{\sigma^2},m+1,\frac{\beta({t,x})^2}{4({\gamma+K({t})})}}\right),
\end{align} 
with
\begin{align}
m&=\frac{2}{\sigma^2}\sqrt{\left({\kappa\theta-\frac{\sigma^2}{2}}\right)^2+2\sigma^2\nu},\\
A&=\kappa^2+2\sigma^2\epsilon,\\
\beta({t,x})&=\frac{\sqrt{Ax}}{\frac{\sigma^2}{2}\sinh\left({\frac{\sqrt{A}t}{2}}\right)},\\
K({t})&=\frac{1}{\sigma^2}\left({\sqrt{A}\coth\left({\frac{\sqrt{A}t}{2}}\right)+\kappa}\right).
\end{align}

If
\begin{equation}
\gamma<-\frac{\sqrt{\kappa^2+2\epsilon\sigma^2}+\kappa}{\sigma^2},
\end{equation}
then the transform is well defined for all $t<t^*$, with
\begin{equation}
t^*=\frac{1}{\sqrt{A}}\log\left({1-\frac{2\sqrt{A}}{\kappa+\sigma^2\gamma+\sqrt{A}}}\right).
\end{equation}
\end{lemma}

\begin{remark}\label{rem3.1}
\textit{Special case}:
when $\epsilon=\nu=\gamma=0$, we have the (non-integral) moments of the process for $\eta < \frac{2\kappa\theta}{\sigma^2}$:
\begin{align*}
\mathbb{E}\left[{X_t^{-\eta}}\right]=&\left({\frac{\kappa}{\sigma^2}}\right)^{\eta}\left({\sinh\left({\frac{\kappa t}{2}}\right)}\right)^{-\frac{2\kappa\theta}{\sigma^2}}\exp\left(\frac{\kappa}{\sigma^2}\left({\kappa\theta t+x-x\coth{\left({\frac{\kappa t}{2}}\right)}}\right)\right)\\
&\times\left({1+\coth\left({{\frac{\kappa t}{2}}}\right)}\right)^{\eta-\frac{2\kappa\theta}{\sigma^2}}\frac{\Gamma\left({\frac{2\kappa\theta}{\sigma^2}-\eta}\right)}{\Gamma\left({\frac{2\kappa\theta}{\sigma^2}}\right)}\ _{1}F_1\left({\frac{2\kappa\theta}{\sigma^2}-\eta,\frac{2\kappa\theta}{\sigma^2},\frac{2\kappa x}{\sigma^2({e^{\kappa t}-1})}}\right)\\
=&G(\kappa,\theta,\sigma,\eta;t,x)
\end{align*}
\end{remark}

\begin{proof}[Proof of Lemma \ref{lem3.2}]
The result follows immediately from Theorem 1 in Grasselli \cite{G2016}, whose proof mainly relies on Lie's classical symmetry method as in Bluman and Kumei \cite{BK1989} and Olver \cite{O1993}.We first note that by standard arguments the expectation is related to the solution of the following symmetrical PDE:
\begin{equation}
u_t=\frac{1}{2}\sigma^2xu_{xx}+f(x)u_x-\left({\frac{\nu}{x}+\epsilon x}\right)u,\quad \epsilon>0,\ \nu>0,
\end{equation}
where $f(x)=\kappa\theta-\kappa x$. The key result in order to find the Lie groups admitted by the PDE states that one should find the invariant surface for the second prolongation of group acting on the $({x;t;u})$-space where the solutions of the PDE lie. Once such equations are solved, one can find the corresponding Lie group admitted by the PDE and thus find a fundamental solution of the PDE by inverting a Laplace transform. Finally, Craddock and Lennox \cite{CL2009} showed the condition under which the fundamental solution is also a transition probability density for the underlying stochastic process. For more details, see Grasselli \cite{G2016}.
\end{proof}

\subsection{3/2 stochastic volatility with jumps in price (3/2-SVJ)}
The Heston model struggles to incorporate the smile in the implied volatilities of short-term index options. Baldeaux and Badran show that unlike for 3/2 model, the implied volatilities are downward sloping, which is not consistent with market data. Finally, they found that 3/2 plus jumps model is able to better fit short-term index implied volatilities while producing more realistic VIX option implied volatilities without a loss in tractability. Setting the volatility parameter $\alpha$ equals to $-\frac{1}{2}$ in equation \eqref{eq2.1}, one arrives at 3/2 plus jump model in equation \eqref{eq2.1}. After some computation process, the expectation in \eqref{eq3.1} can be computed as
\begin{align*}
\mathrm{VIX}_t^2&=\left(\frac{1}{\tau}\mathbb{E}^{\mathbb{Q}}\left[\int_t^{t+\tau}V_t^{-1}dt\right]+2\left(\lambda_1\left(\tilde{\mu}_1-\mu_1\right)+\lambda_2\left(\tilde{\mu}_2-\mu_2\right)\right)\right)\times100^2,t\ge0,\\
&=\left(\frac{1}{\tau}\int_0^{\tau}G(\kappa,\theta,\sigma,1;u,x)du+G_{0}(\mu_1,\mu_2,\tilde{\mu}_1,\tilde{\mu}_2,\lambda_1,\lambda_2)\right)\times100^2.
\end{align*}
%
Hence, the price of a European call option in 3/2 stochastic volatility with jumps is then found by computing the expected payoff directly as:
\begin{align}
\text{C}\left(\mathrm{VIX}_t,K,t,T\right)=&e^{-r\left(T-t\right)}\mathbb{E}^{\mathbb{Q}}\left[\left(\mathrm{VIX}_T-K\right)^{+}|\mathcal{F}_t\right]\notag\\
=&e^{-r\left(T-t\right)}\int_0^{\infty}\left(100\sqrt{\left(\frac{1}{\tau}\int_0^{\tau}Gdu+G_{0}\right)}-K\right)^{+}\times f_{V_T | V_t}^{\mathbb{Q}}(y)dy,\notag\\
\end{align}
while the VIX futures equals:
\begin{align}
\text{F}\left(\mathrm{VIX}_t,K,t,T\right)=&e^{-r\left(T-t\right)}\mathbb{E}^{\mathbb{Q}}\left[\mathrm{VIX}_T|\mathcal{F}_t\right]\notag\\
=&e^{-r\left(T-t\right)}\int_0^{\infty}\left(100\sqrt{\left(\frac{1}{\tau}\int_0^{\tau}Gdu+G_{0}\right)}\right)f_{V_T | V_t}^{\mathbb{Q}}(y)dy
\end{align}

\subsection{Free stochastic volatility and jumps in price (FSV-type)}
To capture the free behavior of volatility, we will relax the power parameter of instantaneous variance, instead of fixing $\alpha$ to be $\frac{1}{2}$ or $-\frac{1}{2}$ as usual. We do not restrict the sign of $\alpha$, letting the data speak as to its direction. 

\begin{theorem}\label{theo3.2}
Let S, V, and $\mathrm{VIX}^2$ be defined by Equation \eqref{eq2.1}, \eqref{eq2.2} and \eqref{eq3.1}. Then 
\begin{equation}
\mathrm{VIX}_t^2=100^2\times\left(H_1+\int_0^{\tau}H_2du\right)
\end{equation}
where
\begin{align}
H_1(\lambda_1,\mu_1,\lambda_2,\mu_2)=&2\left[\lambda_1\left(\tilde{\mu}_1-\mu_1\right)+\lambda_2\left(\tilde{\mu}_2-\mu_2\right)\right],\\
H_2(\kappa,\theta,\sigma,\alpha;u,x)=&\frac{\sigma^{4\alpha}}{\tau\kappa^{2\alpha}}\frac{\Gamma\left({\frac{2\kappa\theta}{\sigma^2}+2\alpha}\right)}{\Gamma\left({\frac{2\kappa\theta}{\sigma^2}}\right)}\left({\sinh{\left({\frac{ku}{2}}\right)}}\right)^{-\frac{2\kappa\theta}{\sigma^2}}\exp\left({\frac{\kappa}{\sigma^2}\left({\kappa\theta u+x-x\coth\left({\frac{\kappa u}{2}}\right)}\right)}\right)\notag\\
&\times\left({1+\coth\left({\frac{\kappa u}{2}}\right)}\right)^{-2\alpha-\frac{2\kappa\theta}{\sigma^2}}\ _{1}F_{1}\left({\frac{2\kappa\theta}{\sigma^2}+2\alpha,\frac{2\kappa\theta}{\sigma^2},\frac{2\kappa x}{\sigma^2({e^{\kappa u}-1})}}\right).
\end{align}

\begin{proof}
From equation \eqref{eq2.3} and \eqref{eq3.1}, we get 
\begin{align}\label{eq3.17}
\frac{\mathrm{VIX}_t^2}{100^2}=&-\frac{2}{\tau}\mathbb{E}^{\mathbb{Q}}\left[\left(-\lambda_1\tilde{\mu}_1-\lambda_2\tilde{\mu}_2\right)\tau-\frac{1}{2}\int_t^{t+\tau}V_s^{2\alpha}ds+\sum_{i=N_1(t)+1}^{N_1(t+\tau)}J_{1i}^{\mathbb{Q}}+\sum_{i=N_2(t)+1}^{N_2(t+\tau)}J_{2i}^{\mathbb{Q}}\bigg|\mathcal{F}_t\right]\notag\\
=&2\left(\lambda_1\tilde{\mu}_1+\lambda_2\tilde{\mu}_2\right)+\frac{1}{\tau}\int_t^{t+\tau}\mathbb{E}^{\mathbb{Q}}\left[V_s^{2\alpha}|\mathcal{F}_t\right]ds-\frac{2}{\tau}\mathbb{E}^{\mathbb{Q}}\left[\sum_{i=N_1(t)+1}^{N_1(t+\tau)}J_{1i}^{\mathbb{Q}}+\sum_{i=N_2(t)+1}^{N_2(t+\tau)}J_{2i}^{\mathbb{Q}}\bigg|\mathcal{F}_t\right]\notag\\
=&2\left[\lambda_1\left(\tilde{\mu}_1-\mu_1\right)+\lambda_2\left(\tilde{\mu}_2-\mu_2\right)\right]+\frac{1}{\tau}\int_0^{\tau}\mathbb{E}^{\mathbb{Q}}\left[V_s^{2\alpha}|\mathcal{F}_t\right]ds\notag\\
=&2\left[\lambda_1\left(\tilde{\mu}_1-\mu_1\right)+\lambda_2\left(\tilde{\mu}_2-\mu_2\right)\right]+\int_0^{\tau}\frac{\sigma^{4\alpha}}{\tau\kappa^{2\alpha}}\frac{\Gamma\left({\frac{2\kappa\theta}{\sigma^2}+2\alpha}\right)}{\Gamma\left({\frac{2\kappa\theta}{\sigma^2}}\right)}\left({\sinh{\left({\frac{ku}{2}}\right)}}\right)^{-\frac{2\kappa\theta}{\sigma^2}}\notag\\
&\times\exp\left({\frac{\kappa}{\sigma^2}\left({\kappa\theta u+x-x\coth\left({\frac{\kappa u}{2}}\right)}\right)}\right)\notag\\
&\times\left({1+\coth\left({\frac{\kappa u}{2}}\right)}\right)^{-2\alpha-\frac{2\kappa\theta}{\sigma^2}}\ _{1}F_{1}\left({\frac{2\kappa\theta}{\sigma^2}+2\alpha,\frac{2\kappa\theta}{\sigma^2},\frac{2\kappa x}{\sigma^2({e^{\kappa u}-1})}}\right)ds\notag\\
=&H_1(\lambda_1,\mu_1,\lambda_2,\mu_2)+\int_0^{\tau}H_2(\kappa,\theta,\sigma,\alpha;u,x)du
\end{align}
and where we used the remark \ref{rem3.1} by setting $\eta$ to be $2\alpha$ and the fact that $V_t$ is Markov process. 
\end{proof}
\end{theorem}

It's worthwhile to note that equation \eqref{eq3.17} is useful as it shows that the distribution of $\mathrm{VIX}_t^2$ can be obtained via the distribution of $V_t$, for $t\ge0$. That's to say, the problem of pricing VIX derivatives is solved to the problem of finding the transition density function for the variance process. By using equation \eqref{eq2.3}, the price of a European call option in free stochastic volatility model is then found by computing the expected payoff directly as:
\begin{align}
\text{C}\left(\mathrm{VIX}_t,K,t,T\right)=&e^{-r\left(T-t\right)}\mathbb{E}^{\mathbb{Q}}\left[\left(\mathrm{VIX}_T-K\right)^{+}|\mathcal{F}_t\right]\notag\\
=&e^{-r({T-t})}\int_0^\infty\left({100\ \sqrt{H_1+\int_0^\tau H_2du}-K}\right)^+\cdot f^{Q}_{V_T | V_t}({y})dy,
\end{align}
while the VIX futures equals:
\begin{align}
\text{F}\left(\mathrm{VIX}_t,K,t,T\right)=&e^{-r\left(T-t\right)}\mathbb{E}^{\mathbb{Q}}\left[\mathrm{VIX}_T|\mathcal{F}_t\right]\notag\\
=&e^{-r({T-t})}\int_0^\infty100\ \sqrt{H_1+\int_0^\tau H_2du}\cdot f^{Q}_{V_T | V_t}({y})dy.
\end{align}

Four volatility considered models can be nested within the Eq.\ref{eq2.1} and Eq.\ref{eq2.2}. Among those, we select four different specifications depending on whether there are restrictions on $\alpha$, $\lambda_1$, or $\lambda_2$. The model specifications considered in this paper are summarized in Table \ref{Table 0}.

\begin{table}
\renewcommand\arraystretch{1.2}
\caption{\\Summary of model specifications}
\label{Table 0}
\centering
\begin{tabular*}{\textwidth}{@{\extracolsep{\fill}}llllllllll}
\hline
Model & & \multicolumn{4}{l}{Description} & & \multicolumn{3}{l}{Constraints}\\ 
\hline
HSV& & \multicolumn{4}{l}{Fixed volatility, no jumps} & & \multicolumn{3}{l}{$\alpha=\frac{1}{2}$, $\lambda_1=0$, and $\lambda_2=0$}\\
3/2-SVJ& & \multicolumn{4}{l}{Fixed volatility, upward and downward jumps} & & \multicolumn{3}{l}{$\alpha=-\frac{1}{2}$}\\
FSV-AJ& & \multicolumn{4}{l}{Free volatility, upward and downward jumps} & & \multicolumn{3}{l}{Not applicable}\\
FSV-DJ& & \multicolumn{4}{l}{Free volatility, downward jumps only} & & \multicolumn{3}{l}{$\lambda_1=0$}\\
\hline
\end{tabular*}
\end{table}

\section{Data description}
Our application looks at the valuation of VIX futures and options, among other volatility derivatives. The market for those derivatives have been explosive growth in trading activity in recent years, as can be see in Fig.~\ref{Fig.1}. The left panel of the figure presents the number of VIX future contract traded increased dramatically from about 0.4 million in 2006 to about 51 million in 2015, and that most of the growth occurred after 2009, likely provoked by the financial crisis. Fig.~\ref{Fig.1}(a) also indicates that VIX futures reflect a demand for a tradable vehicle which can be used to hedge or to implement a view on  volatility. The right panel of the figure shows that the dollar trading volume for the VIX options also increased substantially from around 5 million in 2006 to near 150 million in 2015. 
\begin{figure}[htpb]
\caption{\\Trading volume of VIX futures and options. The left panel shows the numbers of futures contracts traded over time, and the right panel presents the dollar trading volumes for options contracts over time}
\label{Fig.1}
\subfigure[Futures Trading Volume]{\label{Futures Trading Volume}\includegraphics[width=0.45\textwidth]{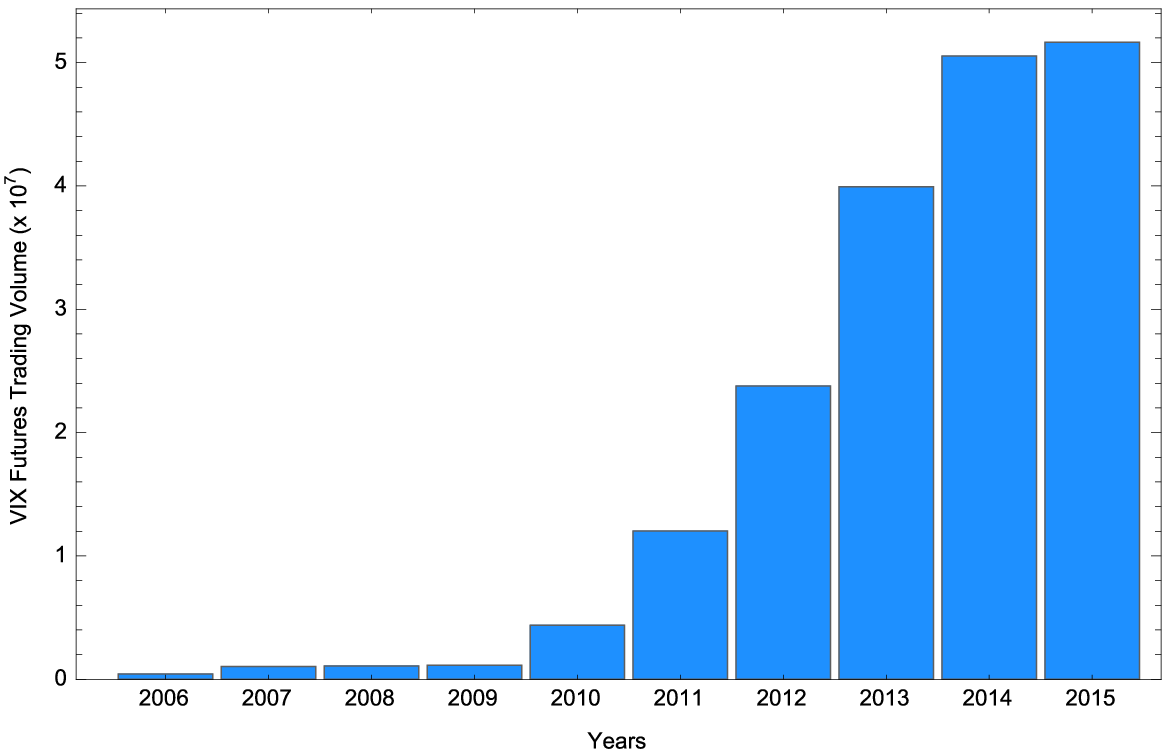}}
\subfigure[Options Trading Volume]{\label{Options Trading Volume}\includegraphics[width=0.45\textwidth]{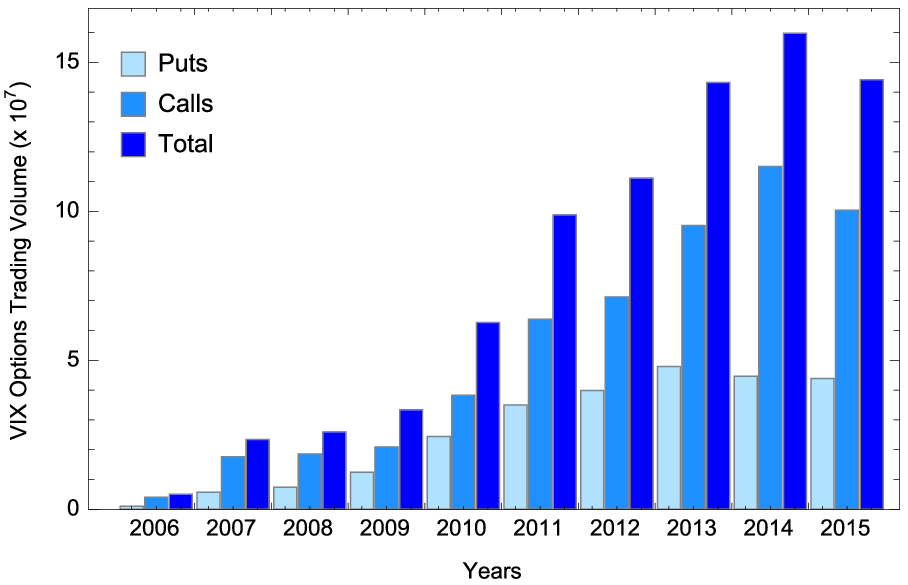}}
\end{figure}
Note that VIX call options are more actively traded than VIX put option, which may be associated with the fact that former can be used as a means of hedging a stock market crash unlike the latter. 

Our data set used in this study comprises VIX index, the corresponding VIX futures and options traded on the CBOE, based on the consideration that the option written on this index is one of the most actively traded contracts. We have tried to include the dates for which data on both VIX futures and options are available. In this framework, the future sample contains 193 futures with a total of 23 trading days and available maturities from 1 day to 268 days while the option sample contains 872 call options, with a total of 10 trading days and available maturities from 1 day to 169 days. Furthermore, our in-sample data employed the delayed market quotes over the period starts on March 1, 2016, and ends on March 20, 2016, with 9 to 10 maturities for each trading day, to calibrate the risk-neutral parameters. Besides, we use those from March 21, 2016 to March 31, 2016 for the out-of-sample test. All details show in Table \ref{Table 2}. On the other hand, in this table, we also break down the data into categories depending on time to maturity $\tau$ and their moneyness. Time to maturity $\tau$ contains short-term contracts with $\tau\le$ one month, middle-term contracts with one month $<\tau\le$ three months and long-term contracts with $\tau>$ three months. Their moneyness is divided into three categories: Out-of-the-Money (OTM) with $d<-0.1$, At-the-money (ATM) with $-0.1\le d\le0.1$ and In-the-money (ITM) with $0.1<d$, where moneyness is defined as $d=\text{VIX}/K$. For each category, we describe the corresponding sample size, average price and average implied volatility. Finally, we just calculate Black-Scholes implied volatility when options are in ITM, ATM and OTM. Option prices are taken from the bid-ask midpoint. To ensure sufficient liquidity and alleviate the influence of price discreetness during the valuation, a filtering scheme is applied to eliminate those inaccurate options by discarding options which the mid price is less than \$2.0 in all data set. Note that the closing hour of the options and the VIX are the same, thus there is no nonsynchronous issue here. Finally, The average one-month US Treasury Bond Rate, 0.05\%, in the period from March 1, 2016 to March 31, 2016 is chosen to be the risk free interest rate.

\begin{table}
\renewcommand\arraystretch{1.1}
\caption{\\Summary statistics for overall sample. The moneyness is defined as $d=\text{VIX}/K$, where K is the strike of option. Short-term contracts are those with no more than one month to maturity, middle-term contracts are those not only more than one month but also no more than three months to maturity and long-term contracts are those more than three months to maturities. Detail descriptions of VIX option data show that the reported numbers are respectively the average option price and the numbers of observations, which are shown in the parenthesis, for the in-sample, out-of-sample. The number of observations for the overall sample and for each moneyness category from Mar. 1, 2016 to Mar. 31,2016. BSIV stands for Black-Scholes implied volatility. OTM, ATM, ITM denote Out-of-the-Money, At-the-Money, In-the-Money options, respectively.}
\label{Table 2}
\centering
\begin{tabular*}{\textwidth}{@{\extracolsep{\fill}}llllll}
\hline
              &             &\multicolumn{4}{c}{Time to maturities}\\ \cline{3-6}
              &             &Short   &   Middle   &  Long &   Total \\
\hline
All futures & No.of futures &    21           &    43      &129       &193   \\
                 & Average price&      17.47    & 19.23    &  20.95  &  20.18 \\
                 &                        &       &     &     &    \\ 
                 
All options& No.of options &     314  & 222    & 336    & 872   \\
                 & Average price &     5.85  &5.41     &4.41     &5.18    \\    
                                  &                        &       &     &     &    \\ 
                                  
OTM options & No.of options &  9     &  61   &   219  & 289   \\
$(d<-0.1)$             & Average price &   2.50    &  2.81   &  3.16   &   3.07 \\
                           & Average BSIV &   0.79    &   0.72  &  0.65   &    0.67\\
                                            &                        &       &     &     &    \\ 
                                                                                        
ATM options       & No.of options &     51      & 48       & 64    & 163   \\
$(-0.1\le d\le0.1)$ & Average price &    3.45   &   4.53  &  5.41   &    4.54\\
                           & Average BSIV &    0.73   &  0.57   &   0.53  &  0.61  \\ 
                                            &                        &       &     &     &    \\ 
                                                                                
ITM options& No.of options &      254 &   113  & 53    &  420    \\
 $(0.1< d)$  & Average price &     6.45  & 7.19    & 8.34    &  6.89  \\     
                    & Average BSIV &    0.87   &  0.57   &   0.46  &  0.73  \\                                                                                             
\hline
\end{tabular*}

\begin{tabular*}{\textwidth}{@{\extracolsep{\fill}}lp{0.9em}lp{0.9em}lp{0.9em}lp{0.9em}lp{0.9em}lp{0.9em}l}
\hline
Detail descriptions of VIX option & &  Total & & \multicolumn{5}{c}{Moneyness $d=\text{VIX}/K$}\\ 
\cline{5-9}
& & & & OTM & &ATM & &ITM \\
& & & & $\le-0.1$ & & $(-0.1,0.1)$ & & $\ge0.1$ \\
\hline
Mar. 1-Mar. 20, 2016 & & \$5.17 & & $\$2.98$ & & $\$4.11$ & & $\$6.77$ \\
                                   & &  (669)    & &      (199)   & &    (119)      & &   (351)\\
Mar. 21-Mar. 31, 2016 & & \$5.24 & & $\$3.27$ & & $\$5.71$ & & $\$7.53$ \\
                                   & &  (203)    & &      (90)   & &    (44)      & &   (69)\\
\hline

\end{tabular*}

\end{table}

\section{Parameter estimates and methodology}
\subsection{Estimation Methodology}
The estimation procedure begins by minimizing the ``loss function", separately for each model and each given day, to derive the risk neutral parameters. As the VIX index and VIX options both contain information about the future dynamics of the VIX index, the following loss functions contain VIX index and VIX options, respectively. $$\text{VIX}Loss=\frac{1}{N_1}\sum_{n=1}^{N_1}\frac{\abs{\text{VIX}_n-\widehat{\text{VIX}}_n}}{\text{VIX}_n}$$ and $$OptionLoss=\frac{1}{N_2}\sum_{n=1}^{N_2}\frac{\abs{C_n-\widehat{C}_n}}{C_n}$$ where $N_{i}, i=1,2$ is the number of the sample data, $\text{VIX}_n$, $\widehat{\text{VIX}}_n$, $C_n$ and $\widehat{C}_n$ represent the market VIX index, model VIX index, the market option price and model option price, respectively. In this section, we adopt a gradient-based minimization algorithm and a local optimization scheme to minimize those loss functions. Note that loss functions are non-linear optimization problem, which can lead to different optimal parameters for different starting parameters. Therefore, the success of the scheme depends on effective choices of the initial parameters. To address this problem, we at first minimize the $\text{VIX}Loss$ function by running calibration 40 times with reasonable parameters which are chosen from the previous literature or GMM estimation. Then, we record a set of VIX-optimum parameters as preparations. Next step, we use those VIX-optimum parameters as initial parameters and then search for news estimates by minimizing $OptionLoss$ function. In all the calibration, the Feller's condition $\frac{2\kappa\theta}{\sigma^2}>1$ and non-explosion condition $\frac{2\kappa\theta}{\sigma^2}>1-\alpha$ are imposed. In addition, we impose additional calculation on updating Hessian matrix approximation at each iteration by using BFGS formula in quasi-Newton method. Hence standard deviation of each parameter can be computed by inverse of Hessian matrix in each iteration. Finally, we will examine whether FSV and jumps help to match the observed term structure of the implied volatility, on average, which will give us a hint as to the pricing performance results discussed in the section that follows.

\begin{table}
\renewcommand\arraystretch{1.2}
\caption{\\Estimated of Risk-Neutral Parameters. Parameters estimates are obtained by minimizing the lossfunction over the period March 1, 2016 to March 20, 2016, followed by its standard error in parenthesis.. The last two rows, labeled $\text{VIX}Loss$ and $OptionLoss$, show the mean absolute percentage errors of VIX ($\text{VIX}Loss$) and the mean absolute percentage errors for all options ($OptionLoss$).}
\label{Table 3}
\centering
\begin{tabular*}{\textwidth}{@{\extracolsep{\fill}}lllllllll}
\hline
Parameters &   & FSV-AJ  &  &FSV-DJ  &  & 3/2-SVJ&  &HSV    \\ 
\hline
$\kappa$   &   &3.8943   &  &3.7029  &  &2.4614 &  & 3.1490  \\
                  &   &(0.010)  &  & (0.010)  &  & (0.009) &  &(0.003)\\
$\theta$   &   & 0.2121  &  &0.2036  &  &47.313 &  & 0.0372  \\
                  &   &  (0.041)&  & (0.039) &  & (0.721) &  &(0.014)\\
$\sigma$    &   & 0.9115  &  &0.8662  &  &-11.0750 &  & 1.0880  \\
                  &   & (0.016)&  & (0.012)&  & (0.118) &  &(0.031)\\
$\alpha$   &   & 1.2156  &  & 1.1575 &  & &  &  \\
                  &   &  (0.143)&  &  (0.116) &  &  &  &\\
$\lambda_1$ &   & 0.0574  &  &  &  & 0.0722  & &   \\
                  &   & (0.037) &  &   &  &(0.002)  &  &\\
$\mu_1$       &   &   0.1125&  &  &  & 0.1518&   &  \\
                  &   &  (0.024)&  &   &  &  (0.032)&  &\\
$\lambda_2$ &   & 0.0648  &  &0.0668  & &0.1203  & &  \\
                  &   &  (0.022)&  & (0.030)   & & (0.026) &  &\\
$\mu_2$ &   &  -0.1232   &  &-0.1233  & &-0.1896  & &  \\
                  &   & (0.021)  &  &  (0.001) &  &  (0.141)&  &\\
$\text{VIX}Loss$&   & 10.20  &  &10.11  &   & 10.81  & & 10.32  \\
$OptionLoss$&   & 6.81  &    &7.21  & &10.65 & &15.07  \\
\multicolumn{9}{@{\extracolsep{\fill}}l}{$\textbf{Notes:}$ Those $\text{VIX}Loss$ and $OptionLoss$ are reported in percentages.}\\                                                                                          
\hline
\end{tabular*}
\end{table}

\subsection{Parameter Estimates and Preliminary Analysis}

Table \ref{Table 3} shows parameters of all models have small standard deviation and hence are stable. Taking a closer look at the estimates, the models with FSV yield very different volatility dynamics than the models without it. Given that estimation results, we'd like to draw the conclusion. First, at last two rows of Table \ref{Table 3}, it presents $\text{VIX}Loss$ and $OptionLoss$ of four models. That can be regarded as a comparison metric. We observe that $\text{VIX}Loss$ and $OptionLoss$ in FSV-type models are more lower than other models, which means FSV models can be more flexible to fit market volatility changes. Specifically, the gradient-based minimization algorithm is used to minimize the objective function $\text{VIX}Loss$ and $OptionLoss$ by running so many iterations until the objective functions are hardly going down. The lower the $\text{VIX}Loss$ and $OptionLoss$ imply the better capability of fitting market VIX index and options. Both information criteria offer decisive rankings of the models: FSV-AJ $>$ FSV-DJ $>$ 3/2-SVJ $>$ HSV. It can be concluded that free stochastic parameter $\alpha$ presents considerable advantages in pricing VIX derivatives.

Second, the FSV-type models ($0.8662\le\sigma_{\text{FSV}}\le0.9115$) have lower volatility of volatility than the HSV and 3/2-SVJ ( $\sigma_{\text{HSV}}=1.0880, \sigma_{\text{3/2-SVJ}}=11.0750$) and hence are more steady. Third, the jump-related parameters ($\lambda_1, \mu_1, \lambda_2, \mu_2$) in FSV-type are significant, suggesting that the underlying asset cannot reject the jump component. On the other hand, let us look at jumps with respective setting. Table \ref{Table 3} contains the frequencies and sizes of upward and downward jumps in the FSV-AJ and FSV-DJ models. Downward jumps take place about 0.065 times per day with an average size of about 0.123, regardless of whether upward jumps are included. It can be seen that in FSV-AJ model, downward jumps have a higher occurrence rate and a larger size than upward jumps; that is, $\lambda_2>\lambda_1$ and $\abs{\mu_2}>\mu_1$. This result also implies that upward and downward jumps should be rejected to be normally distributed together. Overall, allowing for free stochastic parameter $\alpha$ and jumps with respective setting make a large difference in the estimation of the volatility dynamics of the VIX, which in turn will have a large effect on the pricing of VIX derivatives.  

To get a sense of the capability of each model capturing features of the decreasing pattern in the implied volatility, we back out model's implied-volatility series from the Black-Scholes formula by taking the model-determined prices as input and plot those of the ATM call options respectively with different maturities as reflected in Fig.~\ref{Figure 2}. The three panels in Fig.~\ref{Figure 2} shows that all of the four models-HSV (dashed-dotted purple line), 3/2-SVJ (dotted orange line), FSV-DJ (dashed blue long line), and FSV-AJ (dashed blue short line)-do a good job of capturing the decreasing pattern in the observed implied volatility (red dots). The FSV-type's implied volatility pattern fit the market skew quite well across different maturities, while HSV and 3/2-SVJ are no so close to it.

\begin{figure}[htpb]
\caption{\\Graph of VIX's ATM options implied volatility with different maturities. Implied volatilities on March 1, 2016, March 7, 2016 and March 11, 2016 are computed using the market prices and the model-determined prices as inputs to the inverse Black-Scholes formula to obtain the MarketIV and corresponding ModelIV. Only three days are chosen to present.}
\label{Figure 2}
\subfigure[3/1/2016]{\label{3/1/2016}\includegraphics[width=0.32\textwidth]{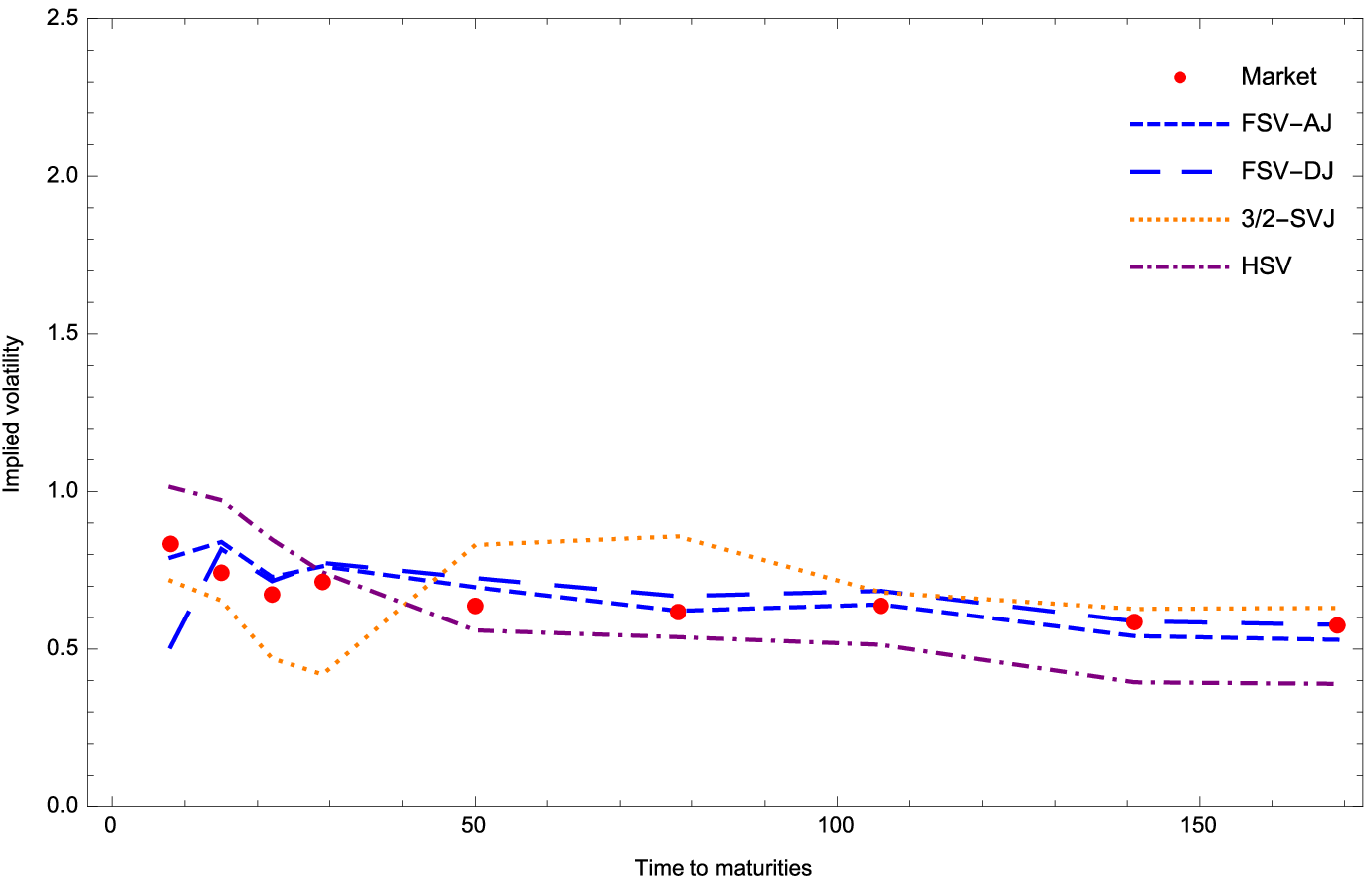}}
\subfigure[3/7/2016]{\label{3/7/2016}\includegraphics[width=0.32\textwidth]{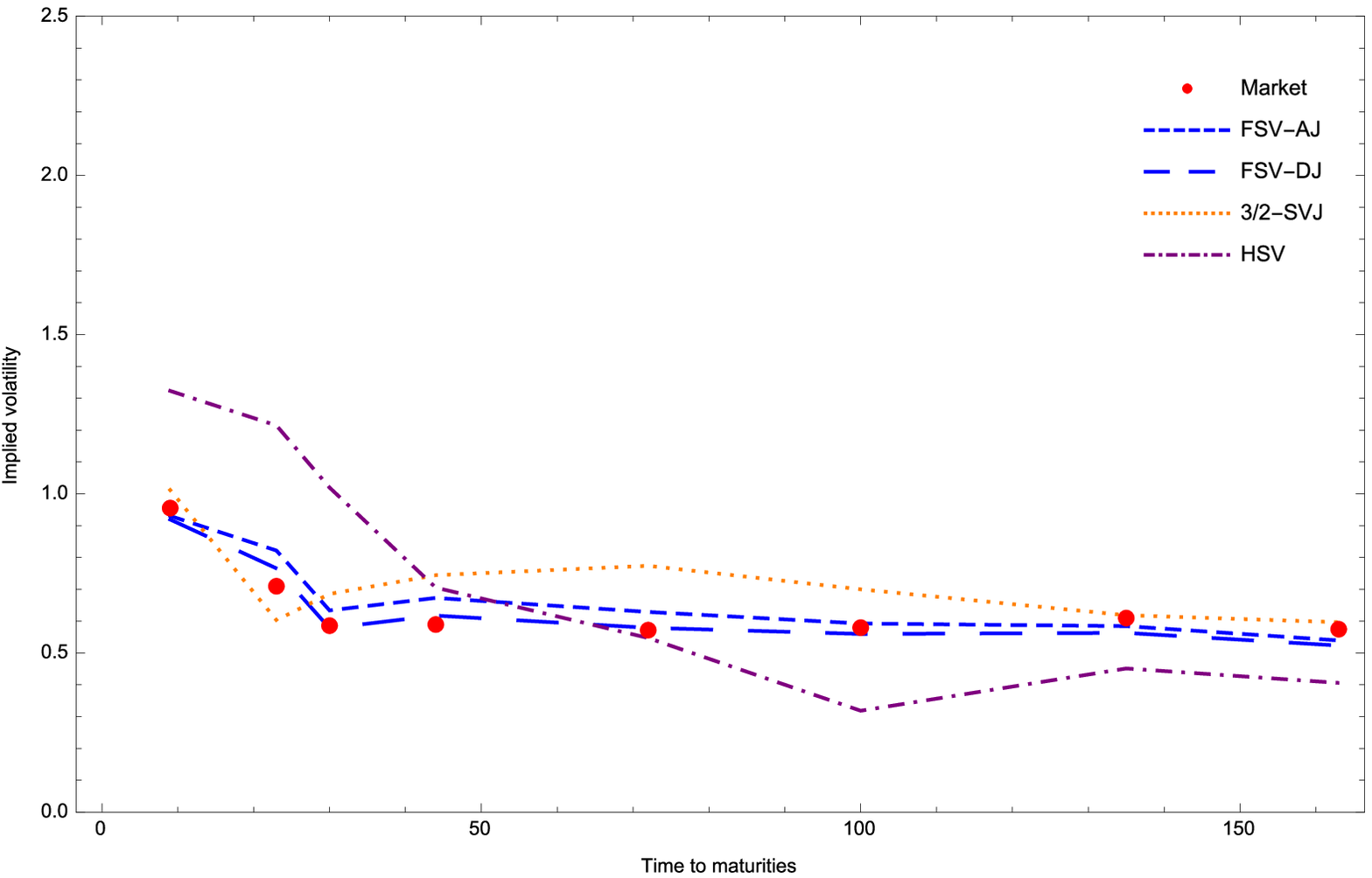}}
\subfigure[3/11/2016]{\label{3/11/2016}\includegraphics[width=0.32\textwidth]{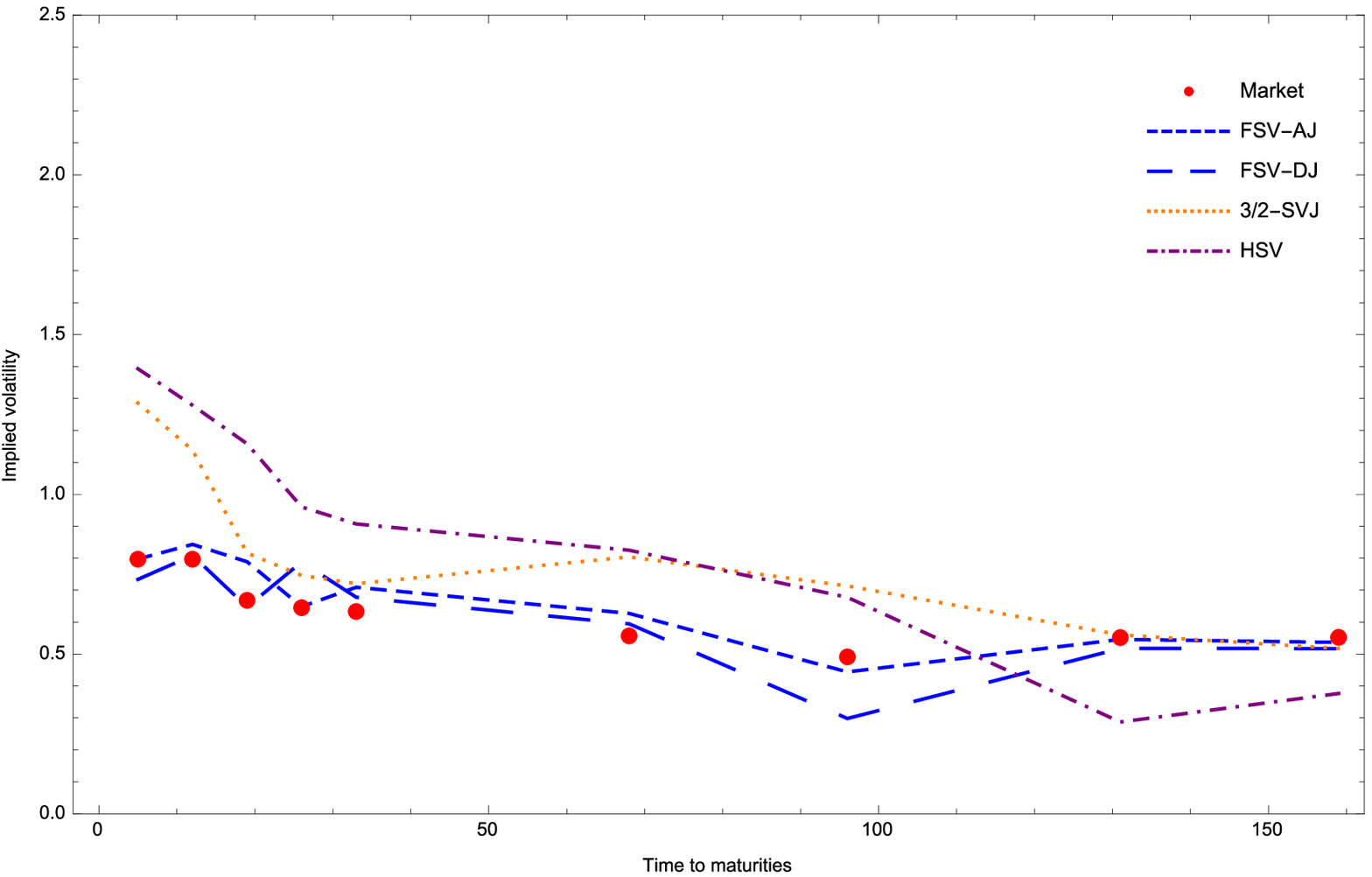}}
\end{figure}

\section{Pricing performance}
Our performance analysis forces on three main comparisons. First, we compare the FSV-type models with other models to test the importance of allowing for free stochastic volatility. Second, we compare the HSV model with the others model to investigate whether the addition of jumps can make an incremental improvement. The third comparison is made between the FSV-AJ model and the FSV-DJ model in order to examine the effects of including upward jumps once free stochastic volatility and downward jumps are included.

To facilitate our analysis, we compare the models by reporting the following three measures of performance:

(1)\quad the ARPE;

(2)\quad the average relative bid-ask error (ARBAE); and

(3)\quad the mean absolute error (MAE):
\begin{align}
\text{ARPE}&=\frac{1}{N}\sum^{N}_{i=1}\frac{\abs{Q_i^{Mid}-Q_i^{Model}}}{Q_i^{Mid}}\\
\text{MAE}&=\frac{1}{N}\sum_{i=1}^{N}\abs{Q^{Model}_i-Q_i^{Mid}}\\
\text{ARBAE}&=\frac{1}{N}\sum^{N}_{i=1}\frac{\text{max}\{(Q^{Model}_i-Q^{Ask}_i)^+,(Q_i^{Bid}-Q_i^{Model})^+\}}{Q_i^{Mid}}.
\end{align}
The ARPE error measure reports the average pricing error in per cent, while the ARBAE measures the average error of model prices that fall outside the bid-ask spread. MAE stands for the sample average absolute difference between the market price and the model price. We only compute the ARBAE measure for the options, as we don't have data available on bid and ask prices for futures. Notice that $Q_i$ is used to denote quotes on calls and futures.

Table \ref{Table 4} presents in-sample pricing performance across the different models. From left to right, the groups of columns show the pricing errors as measured by ARPE, ARBAE and MAE respectively. Each group contains three types of maturities. The main results for all futures and options are shown in the top two panels, and the options are broken down into calls and by moneyness levels in the lower panels. On the whole, FSV-type show the best in-sample performance, being capable of fitting market prices as well as producing volatility skew, and 3/2-SVJ on the other hand, gives competitive performance in consideration of its fewer parameters requirements and good qualifying performance especially in ITM option. To see it, the futures pricing errors and pairwise model comparisons are shown in Panel A of Table \ref{Table 4}, the FSV-type models have in-sample APRE of 0.66-0.67, while the 3/2-SVJ model and HSV model get 0.8 and 3.55 ARPE respectively. Due to both FSV-type and 3/2-SVJ have near lower ARPE, we change to check out the MAE metrics which suggest that the futures pricing performance of the FSV-type model is slightly better than other models in the in-sample test. Given 3/2-SVJ has less parameters than FSV-type, we'd like to draw the conclusion that 3/2-SVJ performed better than FSV-type for in-sample pricing future.

With respect to options, the pricing errors and pairwise model comparisons are shown in Panel B-E of Table \ref{Table 4}. From total statistics of each Panels, the FSV-type show the better in-sample performance than other models. Every error measures between FSV-AJ and FSV-DJ are getting close, which indicate the critical modeling features of free stochastic volatility. The FSV-type models have an in-sample ARPE 6.81-7.21, while the 3/2-SVJ model and HSV model have an in-sample ARPE 10.65 and ARPE 15.07 respectively. In addition, FSV-types model have an in-sample ARBAE of 3.81-4.06, while 3/2-SVJ and HSV models have an in-sample ARBAE of 7.35 and 11.41 respectively larger than FSV-type's. It can be concluded that FSV-types show the lowest average error of model prices falling outside the bid-ask spread. Finally, regarding the reported ARPE value in Panel C, we find that 3/2-SVJ is not appropriate to price OTM options.

{\tiny 
\begin{table}
\renewcommand\arraystretch{1.5}
\caption{\\Pricing performances across different models: In-sample. This table shows the in-sample performance metrics across the different models. The metric ARPE error measure reports the average pricing error in per cent. The metric ARBAE measures the average error of model prices that fall outside the bid-ask spread. The metric MAE stands for the sample average absolute difference between the market price and the model price. The ARPE and ARBAE performance measures are reported in percentages. The moneyness is defined as $d=\text{VIX}/K$ and K is the strike of option contract. Short-term contracts are those with no more than one month to maturity, middle-term contracts are those not only more than one month but also no more than three months to maturity and long-term contracts are those more than three months to maturities. }
\label{Table 4}
\centering
\begin{tabular*}{\textwidth}{@{\extracolsep{\fill}}lllllllllllllll}
\hline
          &    \multicolumn{4}{c}{ARPE} & &\multicolumn{4}{c}{ARBAE}  & &\multicolumn{4}{c}{MAE} \\  \cline{2-5}\cline{7-10}\cline{12-15} 
              &Short   &   Middle   &  Long &   Total & &Short   &   Middle   &  Long &   Total & &Short   &   Middle   &  Long &   Total \\ \hline
\multicolumn{15}{@{\extracolsep{\fill}}l}{\textbf{Panel A: All future}}\\
HSV      & 0.91   &  2.07  &  4.45   & 3.55 &&N/A&N/A&N/A&N/A&&    0.16  & 0.41  &0.94  &0.74\\
3/2-SVJ & 2.32   & 0.85   &  0.55   &0.80 &&N/A&N/A&N/A&N/A&&    0.41   & 0.16  &0.11  &0.16\\
FSV-AJ  &  0.63  & 0.52   &  0.73   &0.66 &&N/A&N/A&N/A&N/A&&    0.11   & 0.10  &0.15  &0.13\\
FSV-DJ  &  0.66  & 0.46   &  0.73   &0.67 &&N/A&N/A&N/A&N/A&&    0.12   & 0.09  &0.16  &0.14\\
\multicolumn{15}{@{\extracolsep{\fill}}l}{\textbf{Panel B: All options}}\\
HSV      & 13.91   &  14.5  &  17.02   & 15.07  &&9.29&11.95&13.94&11.41&&    0.68  & 0.84  &0.84  &0.77\\
3/2-SVJ & 7.97   & 8.07   &  16.02   & 10.65   &&3.95&5.84&12.97&7.35&&    0.37   & 0.35  &0.57  &0.43\\
FSV-AJ  &  6.43  & 8.35   &  6.30   & 6.81      &&2.91&5.90&3.64&3.81&&    0.37   & 0.39  &0.29  &0.35\\
FSV-DJ  &  7.56  & 8.61   &  5.78   & 7.21      &&3.59&6.16&3.28&4.06&&    0.44   & 0.41  &0.26  &0.37\\
\multicolumn{15}{@{\extracolsep{\fill}}l}{\textbf{Panel C: OTM options ($d<-0.1$)}}\\
HSV      & 40.24   &  14.89  &  14.01   & 15.38  &&33.76&12.34&10.56&11.97&&    0.98  & 0.40  &0.45  &0.46\\
3/2-SVJ & 26.6   & 16.01   &  20.07   & 19.53   &&20.60&13.41&16.54&16.08&&    0.65   & 0.41  &0.56  &0.54\\
FSV-AJ  &  13.90  & 13.68   &  6.05   & 7.98      &&7.49&11.07&3.09&4.93&&    0.33   & 0.36  &0.18  &0.22\\
FSV-DJ  &  13.77  & 14.00   &  5.53   & 7.76      &&7.29&11.44&2.78&4.77&&    0.33   & 0.37  &0.16  &0.21\\
\multicolumn{15}{@{\extracolsep{\fill}}l}{\textbf{Panel D: ATM options ($-0.1\le d\le0.1$)}}\\
HSV      & 23.01   &  8.40  &  17.40   & 17.46  &&17.12&5.42&14.64&13.30&&    0.71  & 0.35  &0.87  &0.67\\
3/2-SVJ & 8.42   & 4.00   &  5.09   & 6.23      &&3.83&1.76&2.69&2.94&&             0.26   & 0.17  &0.26  &0.24\\
FSV-AJ  & 6.02  & 6.41   &  4.20   & 5.55      &&2.17&3.60&1.93&2.47&&          0.18   & 0.24  &0.20  &0.20\\
FSV-DJ  & 6.18  & 6.20   &  3.81   & 5.45      &&2.43&3.43&1.64&2.44&&            0.19   & 0.23  &0.18  &0.20\\
\multicolumn{15}{@{\extracolsep{\fill}}l}{\textbf{Panel E: ITM options ($0.1< d$)}}\\
HSV      & 10.91   &  16.85  &  29.45 & 14.09  &&6.70&14.40&27.63&10.46&&    0.66  & 1.27  &2.50  &0.98\\
3/2-SVJ & 7.17   & 5.44   &  10.35   & 7.11      &&3.35&3.42&8.67&3.90&&             0.37   & 0.39  &0.91  &0.43\\
FSV-AJ  & 6.24  & 6.27   &  9.57   & 6.58      &&2.89&4.05&7.76&3.62&&          0.41   & 0.47  &0.83  &0.47\\
FSV-DJ  & 7.63  & 6.08   &  8.96   & 7.55      &&3.70&4.43&7.15&4.20&&            0.49   & 0.50  &0.77  &0.52\\                                                                       
\hline
\end{tabular*}
\end{table}
}

Now that the in-sample fit is increasingly better from HSV, 3/2-SVJ, FSV-DJ and FSV-AJ, one may argue that the outcome can be biased due to the larger number of parameters and the over-fitting to the data. Moreover, a model that performs well in fitting option prices may have poor predictive qualities. Given this concerns, we design the out-of-sample test by using the parameters estimated in Table \ref{Table 3} as inputs to compute the model-based option prices on March 21, 2016 to March 31, 2016 and report the corresponding pricing errors in Table \ref{Table 5}. Let's get a close look at futures. The pricing errors and pairwise model comparisons are shown in Panel A of Table \ref{Table 5}. The FSV-AJ model has an out-of-sample ARPE of 2.76, while the 3/2-SVJ and HSV models have out-of-sample ARPE 3.02 and 5.59 respectively. With respect to options, the pricing errors and pairwise model comparisons are shown in Panel B-H of Table \ref{Table 5}. Totally, the FSV-type models have an out-of-sample ARPE of 8.86-9.97, while the 3/2-SVJ and HSV models have an out-of-sample ARPE of 14.05 and 16.17 respectively. 3/2-SVJ also performs not well in pricing OTM options, and on the other hand, is good at pricing ITM options. According to those results, FSV-type models still show best out-of-sample performance, being capable of generating fewer pricing error in both futures and options, with competitive 3/2-SVJ model in pricing ITM options. Finally, it can be learned from ARBAE that FSV-AJ shows the lowest average error of model prices falling outside the bid-ask spread.

\begin{figure}[htpb]
\caption{\\Graph of the VIX future price as a function of time to maturities (days). VIX futures are computed by four models with respect to 9 different maturities, using the parameters estimated in Table \ref{Table 3}.}
\label{Figure 3}
\subfigure[FSV-AJ]{\label{FSV-AJ}\includegraphics[width=0.24\textwidth]{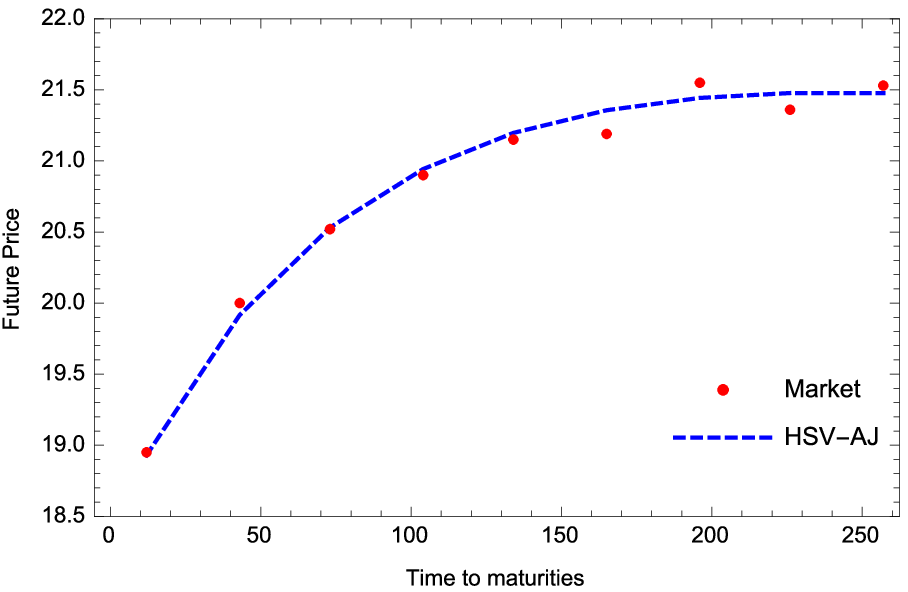}}
\subfigure[FSV-DJ]{\label{FSV-DJ}\includegraphics[width=0.24\textwidth]{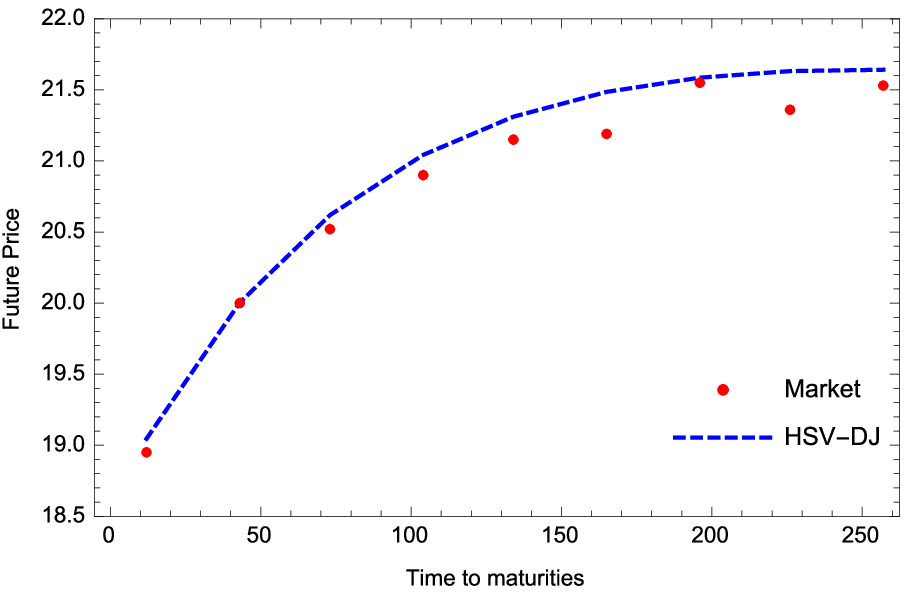}}
\subfigure[3/2-SVJ]{\label{3/2-SVJ}\includegraphics[width=0.24\textwidth]{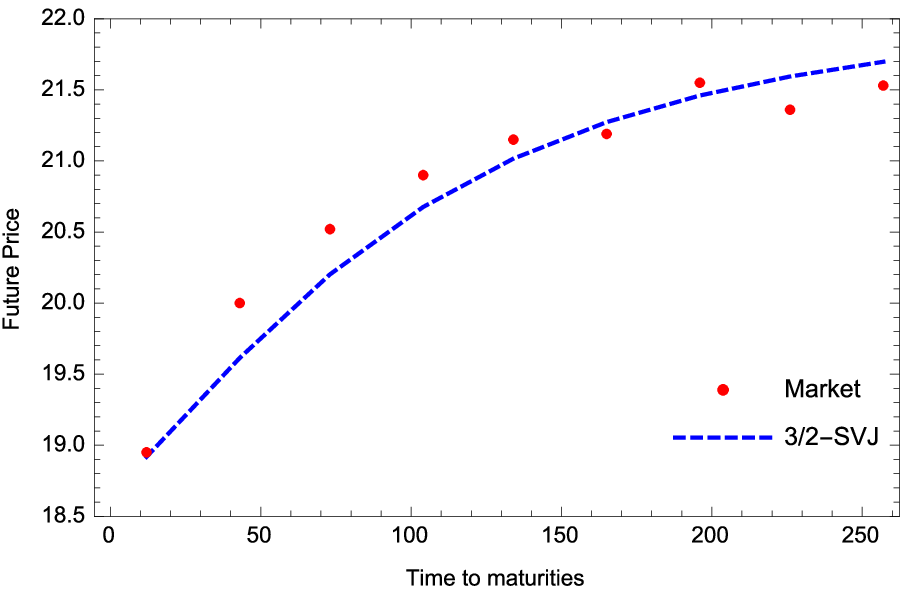}}
\subfigure[HSV]{\label{HSV}\includegraphics[width=0.24\textwidth]{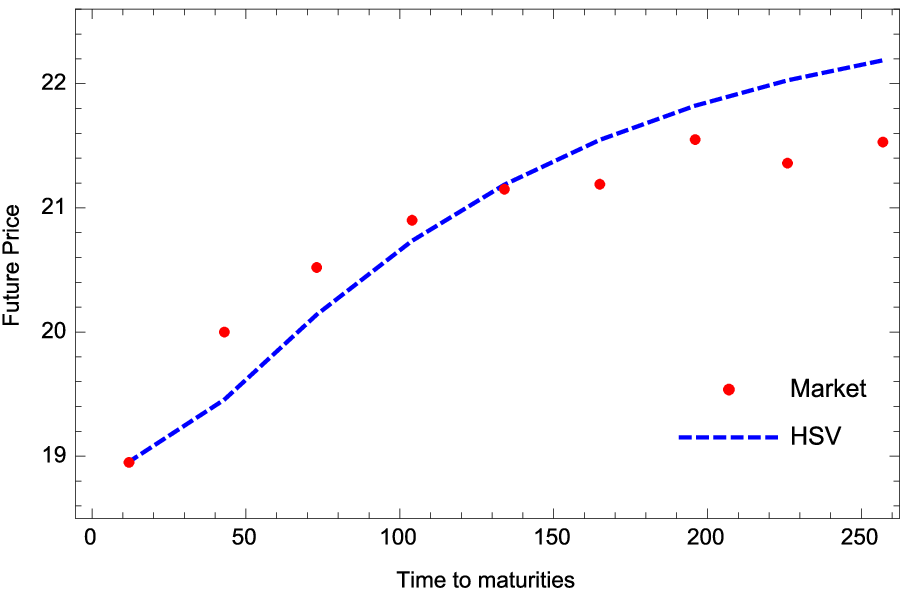}}
\end{figure}

{\tiny 
\begin{table}
\renewcommand\arraystretch{1.5}
\caption{\\Pricing performances across different models: Out-of-sample. This table shows the out-of-sample performance metrics across the different models. The metric ARPE error measure reports the average pricing error in per cent. The metric ARBAE measures the average error of model prices that fall outside the bid-ask spread. The metric MAE stands for the sample average absolute difference between the market price and the model price. The ARPE and ARBAE performance measures are reported in percentages. The moneyness is defined as $d=\text{VIX}/K$ and K is the strike of option contract. Short-term contracts are those with no more than one month to maturity, middle-term contracts are those not only more than one month but also no more than three months to maturity and long-term contracts are those more than three months to maturities. }
\label{Table 5}
\centering
\begin{tabular*}{\textwidth}{@{\extracolsep{\fill}}lllllllllllllll}
\hline
          &    \multicolumn{4}{c}{ARPE} & &\multicolumn{4}{c}{ARBAE}  & &\multicolumn{4}{c}{MAE} \\  \cline{2-5}\cline{7-10}\cline{12-15} 
              &Short   &   Middle   &  Long &   Total & &Short   &   Middle   &  Long &   Total & &Short   &   Middle   &  Long &   Total \\ \hline
\multicolumn{15}{@{\extracolsep{\fill}}l}{\textbf{Panel A: All future}}\\
HSV      & 4.79   &  3.33  &  6.52   & 5.59 &&N/A&N/A&N/A&N/A&&    0.78  & 0.61  &1.31  &1.09\\
3/2-SVJ & 8.50   & 2.95   &  2.09   &3.02 &&N/A&N/A&N/A&N/A&&    1.39   & 0.54  &0.42  &0.56\\
FSV-AJ  &  5.51  & 3.33   &  2.09   &2.76 &&N/A&N/A&N/A&N/A&&    0.94   & 0.63  &0.41  &0.52\\
FSV-DJ  &  2.68  & 2.49   &  4.96   &4.14 &&N/A&N/A&N/A&N/A&&    0.45   & 0.46  &1.00  &0.81\\
\multicolumn{15}{@{\extracolsep{\fill}}l}{\textbf{Panel B: All options}}\\
HSV      & 6.71   &  16.33  &  17.23   & 16.17  &&3.16&13.93&12.58&12.42&&    0.41  & 0.93  &1.03  &0.95\\
3/2-SVJ & 1.24   & 12.64   &  16.52   & 14.05   &&0&10.60&12.46&10.92&&    0.08   & 0.46  &0.55  &0.48\\
FSV-AJ  &  7.14  & 10.82   &  7.82   & 8.86      &&2.86&8.44&4.19&5.65&&    0.45   & 0.48  &0.42  &0.45\\
FSV-DJ  &  12.92  & 13.86   &  7.10   & 9.97    &&8.20&11.42&4.01&7.00&&    0.82   & 0.57  &0.33  &0.45\\
\multicolumn{15}{@{\extracolsep{\fill}}l}{\textbf{Panel C: OTM options ($d\le-0.1$)}}\\
HSV      & N/A   &  21.86  &  10.37   & 12.92  &&N/A&18.71&5.49&8.43&&    N/A  & 0.56 &0.39  &0.43\\
3/2-SVJ & N/A   & 33.34   &  24.83   & 26.72  &&N/A&30.06&19.39&21.77&& N/A   & 0.89  &0.75  &0.78\\
FSV-AJ  &  N/A  & 22.13   &  6.20   & 9.74      &&N/A&18.85&2.22&5.92&&    N/A   & 0.60  &0.20  &0.29\\
FSV-DJ  &  N/A  & 28.96   &  7.69   & 12.42    &&N/A&25.68&4.21&8.98&&    N/A   & 0.80  &0.24  &0.36\\
\multicolumn{15}{@{\extracolsep{\fill}}l}{\textbf{Panel D: ATM options ($-0.1< d<0.1$)}}\\
HSV      & N/A   &  7.83  &  24.49   & 18.05  &&N/A&5.46&19.62&14.15&&    N/A  & 0.41  &1.52  &1.09\\
3/2-SVJ & N/A   & 8.81   &  3.86   & 5.77      &&N/A&6.76&1.88&3.77&&          N/A   & 0.42  &0.22  &0.30\\
FSV-AJ  & N/A  & 7.04   &  8.30   & 7.82      &&N/A&4.63&5.03&4.88&&          N/A  & 0.33  &0.53  &0.45\\
FSV-DJ  &N/A   & 12.31  & 4.83   &  7.72       &&9.75&2.29&5.17&2.44&&            N/A   & 0.59  &0.30  &0.41\\
\multicolumn{15}{@{\extracolsep{\fill}}l}{\textbf{Panel E: ITM options ($0.1\le d$)}}\\
HSV      & 6.71   &  17.24  &  33.00 & 19.21  &&3.16&15.25&29.61&16.54&&    0.41  & 1.38  &2.77  &1.54\\
3/2-SVJ & 1.24   & 3.20   &  3.19   & 2.80      &&0&1.84&1.35&1.33&&             0.08   & 0.24  &0.26  &0.21\\
FSV-AJ  & 7.14  & 6.44   &  13.38   & 8.39      &&2.86&4.56&10.58&5.79&&          0.45   & 0.49  &1.13  &0.65\\
FSV-DJ  & 12.92  & 6.40   &  8.24   & 8.20      &&8.20&4.48&5.84&5.59&&            0.82   & 0.45  &0.70  &0.59\\                                                                      
\hline
\end{tabular*}
\end{table}
}

To further gauge and analyze the effects of including FSV parameter $\alpha$ on pricing futures and options, we now focus on investigation and comparison between FSV-type models with 3/2-SVJ and HSV models. With respect to futures, the FSV-AJ model has in-sample ARPE of 0.66 and has out-of-sample ARPE of 2.76, while both 3/2-SVJ and HSV can't fit the market future price quite well. The in-sample MAE metrics of FSV-type models have been no more than 0.14, which is lower than other models' MAE. Taking a closer look at the short maturities, we find FSV-type models always do better performance, suggesting that the flexible FSV parameter $\alpha$ can take effect on fitting frequently fluctuating the volatility due to short maturities. To get a sense of the capability of each model capturing future price, we plot those of the VIX future value as a function of time to maturities under FSV-AJ, FSV-DJ, 3/2-SVJ and HSV models, and presented in Fig.~\ref{Figure 3}. VIX  futures are computed by using parameters values as given in Table~\ref{Table 3}. Here time to maturities presented in Fig.~\ref{Figure 3} are not annualized. All of this indicate that the futures pricing performance of FSV-type models are better than that of the 3/2-SVJ and HSV models, showing the importance of including free stochastic volatility.

With respect to options, we break down option pricing errors into maturities and several moneyness defined as $d=\text{VIX}/K$ which $K$ is strike of option contract. The FSV-type models have in-sample ARPE of 6.81-7.21 and MAE of 0.35-0.37, while 3/2-SVJ model has in-sample ARPE of 10.65 and MAE of 0.43. HSV model perform not quite well in pricing options, having in-sample ARPE of 15.07 and MAE of 0.77. Taking our analysis one step further. Based on Panel C of in-sample and out-of-sample presented above, we conclude that FSV-type models greatly improve the pricing of OTM options, especially comparing with 3/2-SVJ model. To see ITM options presented in Panel E of Table \ref{Table 5}, however, all the models except 3/2-SVJ generate larger percentage errors in the ITM options of out-of-sample test, which shows that 3/2-SVJ is quite competent in ITM options pricing with the advantage of its less parameters and simplicity. To get a sense of the capability of each model capturing option price, we plot those of the VIX option value as a function of 38 strikes, from 10 to 70, under time to maturities 28 days, and presented in Fig.~\ref{Figure 10}. FSV-type outperform other models in OTM options while 3/2-SVJ indeed does good job in ITM options. After comparing all in-sample ARPE with out-of-sample ARPE in each models, FSV-type models perform more stable in pricing futures and options than other models. To sum up, including free volatility parameter $\alpha$ in our model plays a crucial role in fitting the prices of futures and options.

\begin{figure}[htpb]
\caption{\\Graph of the VIX option price as a function of strikes. VIX options are computed by four models with respect to 38 strikes, using the parameters estimated in Table \ref{Table 3}. Here the strikes are divided by 100.}
\label{Figure 10}
\subfigure[FSV-AJ]{\label{FSV-AJ}\includegraphics[width=0.24\textwidth]{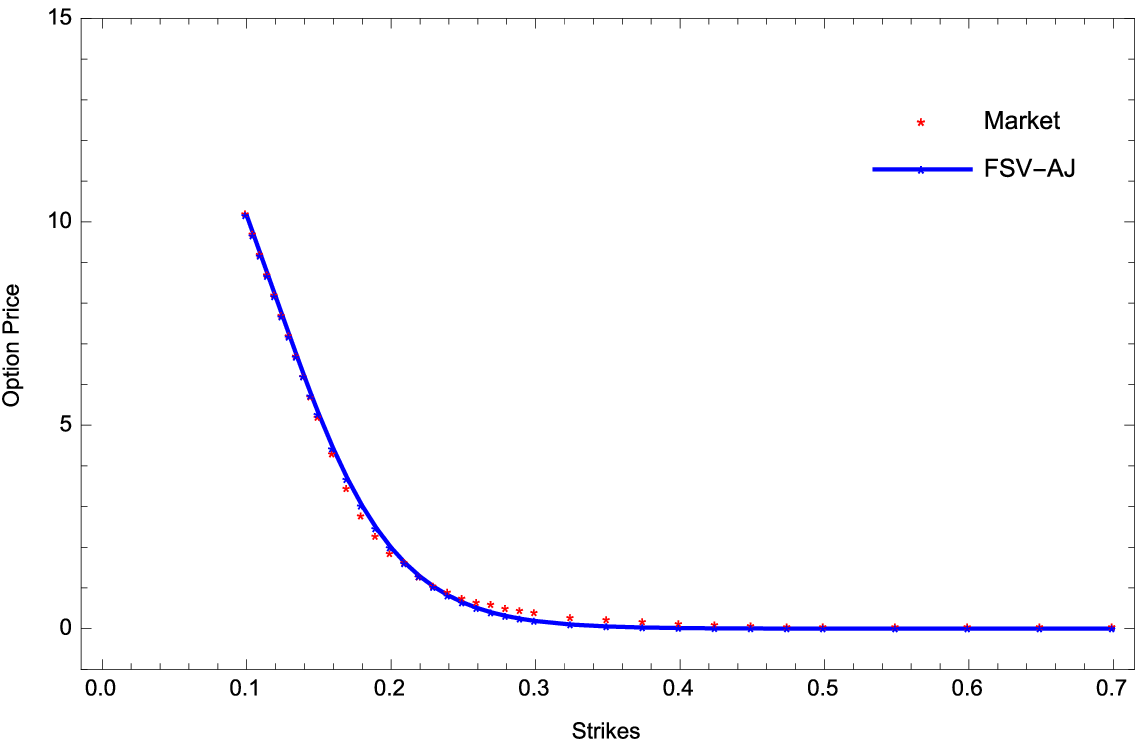}}
\subfigure[FSV-DJ]{\label{FSV-DJ}\includegraphics[width=0.24\textwidth]{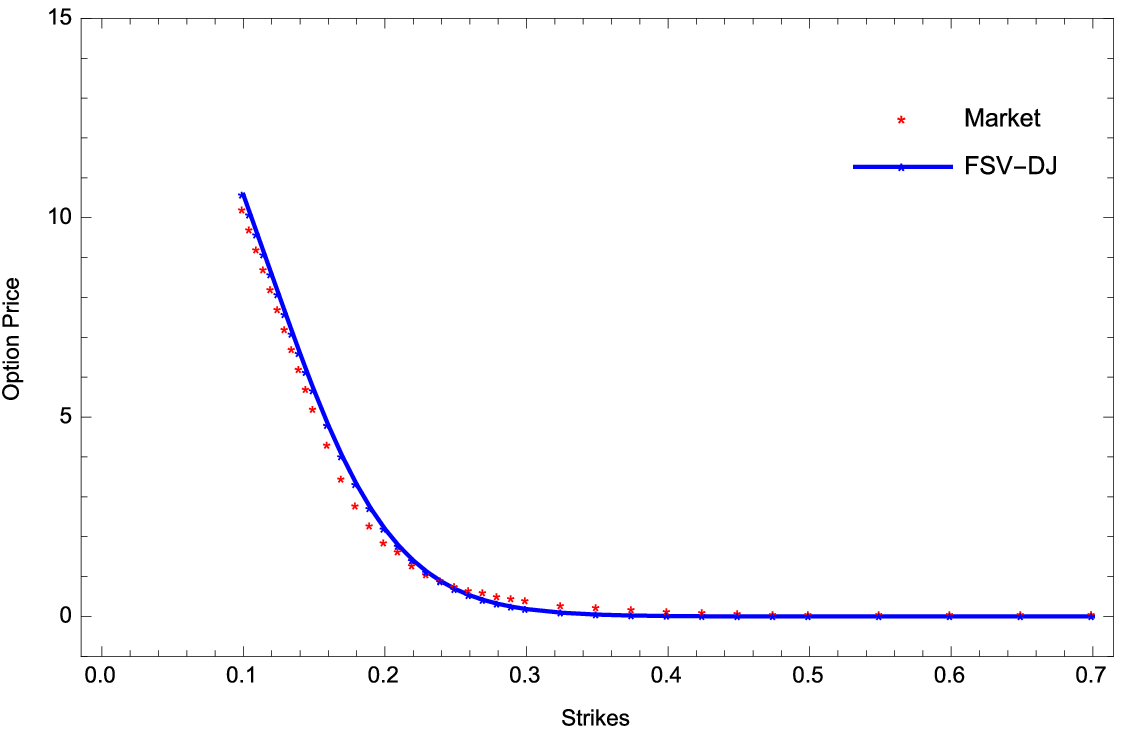}}
\subfigure[3/2-SVJ]{\label{3/2-SVJ}\includegraphics[width=0.24\textwidth]{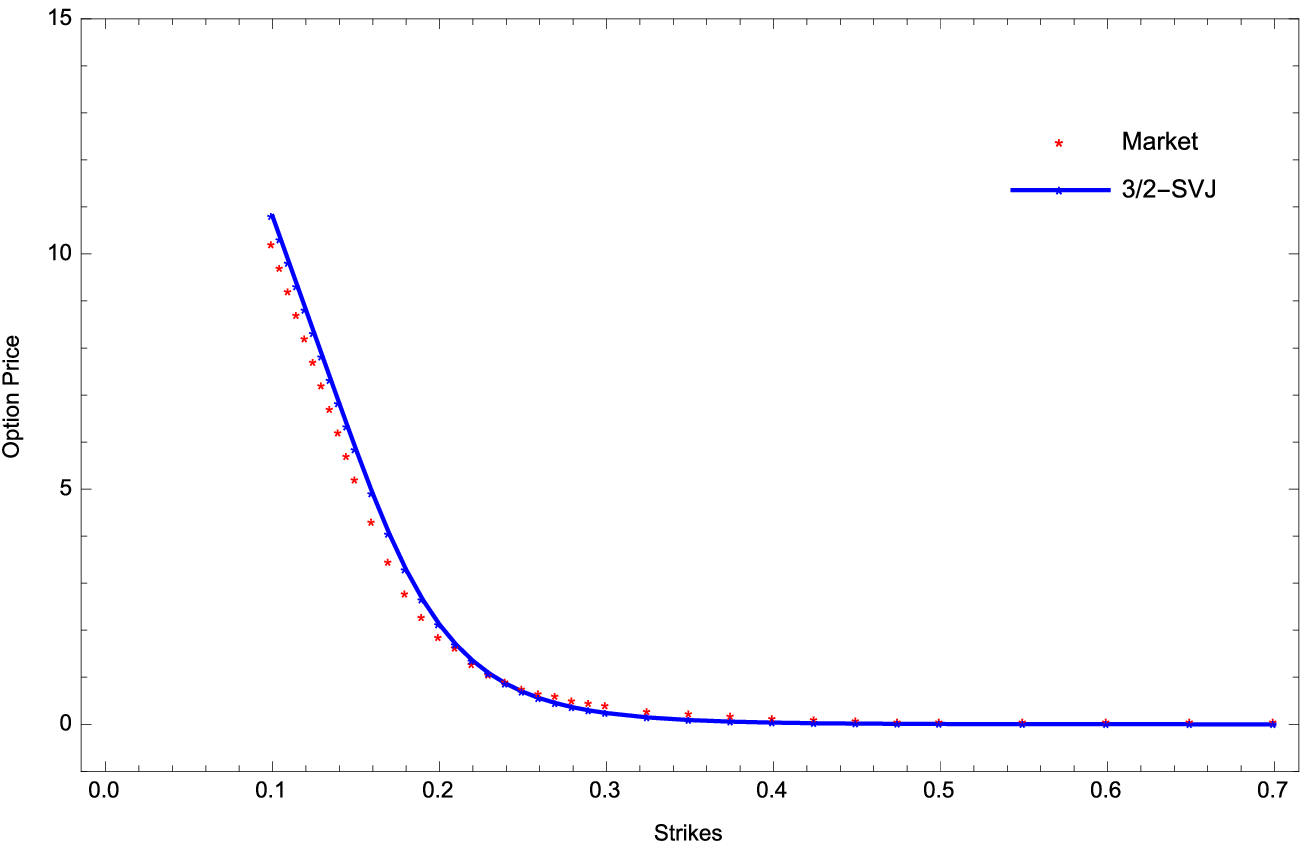}}
\subfigure[HSV]{\label{HSV}\includegraphics[width=0.24\textwidth]{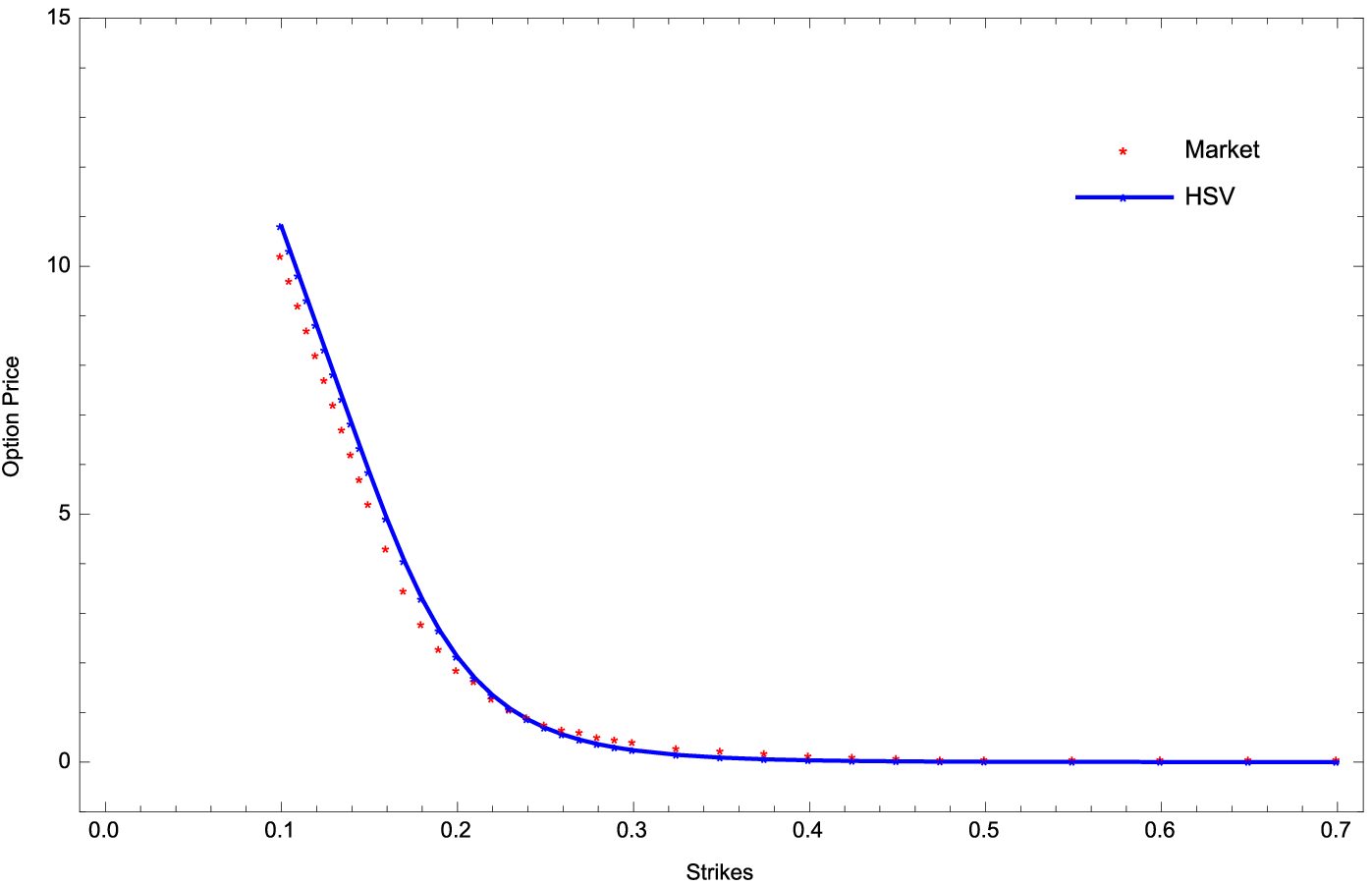}}
\end{figure}

Lastly, we evaluate the incremental effect of upward and downward jumps on the pricing of futures and options. We now get close look at FSV-AJ and FSV-DJ pricing errors. Panel A to E of Table \ref{Table 4} and \ref{Table 5} show the in-sample and out-of sample ARPE and MAE. On the whole, FSV-AJ shows the better in-sample and out-of-sample performance in futures and OTM options. Heuristically, it seems that including upward jump in FSV-AJ contribute VIX to going down to make OTM options back to ITM options and hence FSV-AJ does good job in OTM options pricing. It also can be seen in Fig.~\ref{Figure 10} that FSV-AJ is more closed to market price at lower strike ahead of 20. Although the FSV-AJ model greatly improves the pricing of OTM call options in out-of-sample tests, there is no significant pricing difference between the FSV-AJ and FSV-DJ with respect to the pricing of in-sample futures and option and out-of-sample ITM options. It can be concluded that including upward jumps cannot be rejected to the extent free stochastic volatility and downward jumps are already accounted for.

\section{Conclusion}
In this paper, we have introduced an efficient and flexible model, FSV model, that is an extension of the famous Heston model and 3/2 plus jumps model in a unified framework by keeping their analytical tractability, in the sense that we first derive and establish the quasi-closed form solution for future and call option price. The models have three innovative features which are wholly different from those previous appearing in the literature. First, we relax the power parameter $\alpha$ of instantaneous variance of CIR process, instead of fixing $\alpha$ to be some fixed value as usual. Hence we set free FSV parameter $\alpha$, letting the data speak as to its direction. Furthermore, we also provide non-explosion condition $\frac{2\kappa\theta}{\sigma^2}>1-\alpha$ and prove the discounted price stock is a true martingale. Second, our paper first testified that FSV-type models can be more closed to S\&P500 by using GMM technique as test. Results show that $\alpha$ of the indices can imply different volatility fluctuation in different periods, in other word, another subclass of FSV-type is preferred and newly added $\alpha$ can help to identify the equity distribution. Third, we separate out the roles of upward and downward jumps in equity, instead of just getting them together in a normal distribution, to better consider the reality that investors react differently to good and bad shocks. 

According to the in-sample measure and out-of-sample forecast, adding the FSV parameter $\alpha$ and jumps features to the model can greatly improve its performance, especially in pricing OTM options. Model with jumps outperforms in pricing performance and we find that upward jumps play more crucial feature in pricing OTM option. On the other hand, although all of the models considered do a nice job of capturing the decreasing pattern in the implied volatility, only FSV-AJ and FSV-DJ models successfully match the flat pattern in the implied skewness, consistent with the results that FSV parameter $\alpha$ and jumps are both important for pricing the VIX derivatives.   

Finally, there are several works remained which have not been discussed in this paper and we hope to figure them in the future. For example, to hedge European-style options in discrete time, Fourier cosine series expansions and characteristic function of underlying asset return process are hoped to provide. Furthermore, whether FSV-type models can accurately price S\&P500 option and realized-variance derivatives is still unknown. Adding jumps in instantaneous variance process should be worth exploring. Those topics also deserve more future research and will be the subject of future work.

\begin{appendix}\label{appendix1}
\section{}
\label{app}
To derive $\mathbb{E}^{\mathbb{Q}}[\varepsilon_{t+1}]$ and $\mathbb{E}^{\mathbb{Q}}[\varepsilon_{t+1}^{2}]$, we split the Brownian motion $W$ into $Z$ and its orthogonal part $Z^{\bot}$ and obtain
\begin{align*}
\varepsilon_{t+1}&=V_{t}^{\alpha} W(\Delta t)\\
&= V_{t}^{\alpha} \left[\rho Z(\Delta t)+\sqrt{1-\rho^2}Z^{\bot}(\Delta t)\right].
\end{align*}
Taking expectation to two sides, it follows that
\begin{align*}
\mathbb{E}^{\mathbb{Q}}[\varepsilon_{t+1}]&=\rho\mathbb{E}^{\mathbb{Q}}\left[V_t^{\alpha}Z(\Delta t)\right]+\sqrt{1-\rho^2}\mathbb{E}^{\mathbb{Q}}\left[V_t^{\alpha}Z^{\bot}(\Delta t)\right]\\
&=\rho\mathbb{E}^{\mathbb{Q}}\left[V_t^{\alpha}Z(\Delta t)\right]+\sqrt{1-\rho^2}\mathbb{E}^{\mathbb{Q}}\left[V_t^{\alpha}\right]*\mathbb{E}^{\mathbb{Q}}\left[Z^{\bot}(\Delta t)\right]\\
&=0
\end{align*}
The first term in the last equation equals 0 due to the independent and independent increments of Brownian motion. On the other hand, the second moment of $\varepsilon_{t+1}$ is explicitly known as 
\begin{align*}
\mathbb{E}^{\mathbb{Q}}[\varepsilon^{2}_{t+1}]&=\mathbb{E}^{\mathbb{Q}}\left[^2 V_{t}^{2\alpha}W^2(\Delta t)\right]\\
&=\mathbb{E}^{\mathbb{Q}}[V_{t}^{2\alpha}\left(\rho^2Z^2(\Delta t)+(1-\rho^2)(Z^{\bot}(\Delta t))^2\right)+2\rho\sqrt{1-\rho^2}V_{t}^{2\alpha}Z(\Delta t)Z^{\bot}(\Delta t)]\\
&=\mathbb{E}^{\mathbb{Q}}[V_{t}^{2\alpha}]\Delta t+2\rho\sqrt{1-\rho^2}\mathbb{E}^{\mathbb{Q}}[V_{t}^{2\alpha}Z(\Delta t)]*\mathbb{E}^{\mathbb{Q}}[Z^{\bot}(\Delta t)]\\
&=\mathbb{E}^{\mathbb{Q}}[V_{t}^{2\alpha}]\Delta t\\
&=\Delta t \frac{\Gamma\left[2\left(\alpha+\frac{\kappa\theta}{\sigma^2}\right)\right]}{\Gamma{\left[\frac{2\kappa\theta}{\sigma^2}\right]}}\left(\frac{\sigma^4}{4}\right)^{\alpha}e^{-2\kappa\left(t\left(\alpha+\frac{\kappa\theta}{\sigma^2}\right)+\frac{V_0}{\left(e^{t\kappa}-1\right)\sigma^2}\right)}\notag\\
&\quad\times\left(\frac{\kappa}{1-e^{-t\kappa}}\right)^{\frac{2\kappa\theta}{\sigma^2}}\left(\frac{\kappa}{-1+e^{t\kappa}}\right)^{-2\alpha-\frac{2\kappa\theta}{\sigma^2}} \ _{1}F_1\left[2\left(\alpha+\frac{\theta\kappa}{\sigma^2}\right),\frac{2\kappa\theta}{\sigma^2},\frac{2V_{0}\kappa}{\left(-1+e^{t\kappa}\right)\sigma^2}\right]
\end{align*} 
where we used the remark \ref{rem3.1} by setting $\eta$ to be $2\alpha$ and the fact that $V_t$ is Markov process.
\end{appendix}


\begin{thebibliography}{99}




\bibitem{B2000}
Bates, D. S. (2000). Post-'87 crash fears in the S\&P 500 futures option market, \textit{Journal of Econometrics}, 94(1--2): 181--238.

\bibitem{BB2014}
Baldeaux, J. and Badran, A. (2014). Consistent modelling of VIX and equity derivatives using a 3/2 plus jumps model. \textit{Applied Mathematical Finance}, 21(4): 299--312.


\bibitem{BK1989}
Bluman, G. and Kumei, S. (1989). \textit{Symmetries and differential equations}. Applied Mathematical Sciences, 81. Springer-Verlag, New York, xiv+412 pp.

\bibitem{BL1978}
Breeden, D. T. and Litzenberger, R. H. (1978). Prices of state-contingent claims implicit in option prices. \textit{Journal of business}, 51(4): 621--651.




\bibitem{CBOE2003}
CBOE (2003). VIX: CBOE volatility index. Available at: \url{www.cboe.com/micro/vix/vixwhite.pdf}.


\bibitem{CIR1985}
Cox, J. C., Ingersoll, J. E, and Ross, S. A. (1985). A theory of the term structure of interest rates. \textit{Econometrica}, 53(2): 385--407.

\bibitem{CL2009}
Craddock, M. and Lennox, K. A. (2009). The calculation of expectations for classes of diffusion processes by Lie symmetry methods. \textit{The Annals of Applied Probability}, 19(1): 127--157.



\bibitem{D2012}
Drimus, G. G. (2012). Options on realized variance by transform methods: a non-affine stochastic volatility model. \textit{Quantitative Finance}, 12(11): 1679--1694.

\bibitem{DY2010} 
Duan, J. C. and Yeh, C. Y. (2010). Jump and volatility risk premiums implied by VIX. \textit{Journal of Economic Dynamics and Control}, 34(11): 2232--2244.




\bibitem{GM2013}
Goard, J. and Mazur, M. (2013). Stochastic volatility models and the pricing of VIX options. \textit{Mathematical Finance}, 23(3): 439--458.

\bibitem{G2016}
Grasselli, M. (2016). The 4/2 stochastic volatility model: a unified approach for the Heston and the 3/2 model. To appear in \textit{Mathematical Finance}.

\bibitem{GL1996}
Grünbichler, A. and Longstaff, F. A. (1996). Valuing futures and options on volatility. \textit{Journal of Banking \& Finance}, 20(6): 985--1001.

\bibitem{H1982}
Hansen, L. P. (1982). Large sample properties of generalized method of moments estimators. \textit{Econometrica}, 50(4): 1029--1054


\bibitem{H1993}
Heston, S. L. (1993). A closed-form solution for options with stochastic volatility with applications to bond and currency options. \textit{Review of Financial Studies},   6(2): 327--343.

\bibitem{C1992}
Chan, K. C., Karolyi, G. A., Longstaff, F. A., and Sanders, A. B. (1992). An empirical comparison of alternative models of the short-term interest rate. \textit{The Journal of Finance}, 47(3): 1209--1227.

\bibitem{L2000}
Lewis, A. L. (2000). \textit{Option Valuation under Stochastic Volatility}. Finance Press, Newport Beach, CA, viii+350 pp.

\bibitem{LZ2013}
Lian, G. H. and Zhu, S. P. (2013). Pricing VIX options with stochastic volatility and random jumps. \textit{Decisions in Economics and Finance}, 36(1): 71--88.


\bibitem{O1993}
Olver, P. J. (1993). \textit{Applications of Lie groups to differential equations. Second edition}. Graduate Texts in Mathematics, 107. Springer-Verlag, New York, xxviii+513 pp.

\bibitem{P2002}
Pan, J., (2002). The jump-risk premia implicit in options: evidence from an integrated time-series study. \textit{Journal of Financial Economics}, 63(1): 3--50

\bibitem{P2016}
Park, Y. H. (2016). The Effects of asymmetric volatility and jumps on the pricing of VIX derivatives. \textit{Journal of Econometrics}, 192(1): 313--328



\bibitem{S2008}
Sepp, A. (2008). VIX option pricing in a jump-diffusion model. \textit{Risk Magazine}: 84--89.


\bibitem{W1985}
Whitney, K. N. (1985). Generalized method of moment specification testing. \textit{Journal of Econometric}, 29: 229--256

\bibitem{ZZ2006}
Zhang, J. E. and Zhu, Y. (2006). VIX futures. \textit{Journal of Futures Markets}, 26(6): 521--531.





\end{thebibliography}
\end{document}